\documentclass[runningheads]{llncs}
\usepackage[T1]{fontenc}
\usepackage{fullpage}
\usepackage{graphicx}

\usepackage{amssymb, amsmath, latexsym, amsfonts}
\usepackage{algorithm}
\usepackage[noend]{algpseudocode}
\usepackage{graphicx}
\usepackage{natbib}

\newcommand{\RR}{\mathbb{R}}

\newcommand{\eps}{\varepsilon}

\newcommand{\E}{\mathbb{E}}

\usepackage[colorlinks=true, allcolors=blue]{hyperref}
\usepackage{subcaption}

\usepackage{tikz}
\usepackage{tikz-cd}
\usetikzlibrary{decorations.markings,arrows.meta}
\usepackage{amsmath}
\newsavebox{\inhibit}

\begin{document}
\savebox{\inhibit}{A \begin{tikzpicture} \draw[-{Bar[width=6pt]}](0,0) -- (0.85em,0);\end{tikzpicture} B}

\title{Approximating a gene regulatory network from non-sequential data}
%
%
\author{Pratik Worah\inst{1}\and
Cliff Stein\inst{2}}
%
%
\institute{NYU and Google Research \and
Columbia U. and Google Research}
\maketitle              
\begin{abstract}
Given non-sequential snapshots from instances of a dynamical system, we design a compressed sensing based algorithm that reconstructs the dynamical system. On the theoretical side, we show that: (1) successful reconstruction is possible under the assumption that we can construct an approximate clock from a subset of the coordinates of the underlying system, and (2) computing the minimal Lyapunov exponent of the dynamical system, where the minimum is taken over all subsets of coordinates of the dynamical system, equates to computing a min-max equilibrium. We design an efficient randomized algorithm for computing the above equilibrium.

As an application of our theoretical results, we reconstruct the underlying dynamical system from publicly available RNA-seq data to: (1) predict the underlying gene regulatory networks (as opposed to individual genes) that may help differentiate between metastatic vs non-metastatic breast cancer (and also colorectal cancer), and (2) identify candidate genes that could be used as target biomarkers for basket trials. In particular, our in silico analysis suggests that RORC agonists, which are already used in colorectal cancer therapies, may be worth investigating for breast cancers.

\end{abstract}
\clearpage
\section{Introduction}

Reconstructing a noisy linear dynamical system from sequential observations to predict its next step is a basic problem in filtering theory. In this paper, we consider the case when the order of observations has been partially lost. Such a situation arises often in biological samples, where although one can collect (aggregate) information about cells, it is impossible to know perfectly where the cells are in their cell cycle relative to each other. In particular, the relative abundance of bulk RNA-seq data over single cell RNA-seq data in clinical settings was a motivation for formulating the problem. Note that the latter can be temporally ordered in a reasonably error free manner using latent time, but it is not obvious whether the samples in the former can be ordered as well.
Our contributions may be viewed from a theoretical lens -- abstract algorithms with applications to computational biology, or from a genomic lens -- isolating underlying pathological pathways and common biomarkers from various cancers for designing basket trials.


\subsection{Theoretical contribution}

Consider a linear approximation of a dynamical system: $\dot{x}=Bx+\dot{\eta}$, where $\dot{x}$ denotes the time derivative of $x\in\RR^n$, $B\in\RR^{n\times n}$ is the matrix representing the interactions between the variables, and $\dot{\eta}$ is a white noise term (assumed not too large). If we had a set of ordered samples $x(t),x(t+h),x(t+2h),...$ then we could recover the matrix $B$ from these samples by first computing the discrete time derivative to recover $\dot{x}$, and then minimizing an appropriate loss ($\ell_1$ loss, for example). If the samples are not ordered, or if we don't have access to the timestamps $t,t+h,t+2h,..$, then it is not possible to recover the interaction matrix $B$. However, it is natural to ask, if the sample ordering is known up to some small error, for example a small fraction of the ordered samples have been flipped, then can we recover the dynamical system?

We show that one can still recover $B$ by minimizing the $\ell_1$ loss. We denote this access to a small error sample ordering as access to an approximate clock.\footnote{The terminology comes from our application to genomics, where certain proteins act as approximate clocks.} Under the assumption that we have access to an approximate clock, Algorithm~\ref{algmain} in Section~\ref{ovw:sec} allows us to approximately reconstruct the interaction matrix $B$ from the samples. Moreover, we bound the error in approximation of $B$ from Algorithm~\ref{algmain} using Theorems~\ref{uniqueness:thm} and~\ref{clock:thm}. The error bound in Theorem~\ref{clock:thm} is in terms of the amount of error from the actual sample ordering in the approximate clock. Thus we rely on a combination of techniques from stochastic calculus and compressed sensing to show that the second moment of increments of the underlying dynamical system is robust under errors in ordering (see Theorem~\ref{clock:thm}) and it suffices to reconstruct the underlying sparse dynamical system (Theorem~\ref{uniqueness:thm}). This combination of techniques should be novel and interesting.\footnote{Although there is some novelty in our modification of the basis pursuit algorithm to enforce row sparsity of matrices (see Theorem~\ref{uniqueness:thm}), instead of sparsity of vectors,
we do not claim significant theoretical novelty over the proof techniques of~\cite{dirk} in terms of necessary and sufficient criteria for reconstruction of sparse signals.} The main utility of compressed sensing methods in Algorithm~\ref{algmain} is to ensure that reconstruction can be successful, even with small sample sizes (relative to underlying dimension, i.e., $n$), which is often the case with RNA-seq data -- far fewer patients than number of genes sequenced. Our main application involves reconstructing dynamical systems and predicting pathways responsible for metastasis in breast and colorectal cancers (see Figure~\ref{paths:fig}, and Sections~\ref{gene5:sec} and~\ref{gene2:sec} for details).\\

\noindent Suppose we successfully reconstruct interaction matrices $B$ and $B'$ of two different dynamical systems using Algorithm~\ref{algmain}, a natural question is: Do the two given linear dynamical systems evolve differently over a long period of time? More concretely, is there a subset of coordinates $S$ ($S\subset[n]$) such that the orbit corresponding to $B'$ projected to $S$ converges to some equilibrium point, but the orbit of $B$ does not converge to any point? 

One metric to compare similarity of trajectories is the Lyapunov exponent, which measures whether the long term trajectory of a dynamical system diverges from, or converges to, an equilibrium point. Thus if we find coordinates $S\subset [n]$ such that the Lyapunov exponent corresponding to $B$ (after projection to $S$) is positive and large, but the Lyapunov exponent corresponding to $B'$ (after projection to $S$) is significantly smaller than that of $B$ on $S$, then we know that the two dynamical systems have very different long term trajectories on the $S$ coordinates. Since the computation of a Lyapunov exponent for any linear system corresponds to computing the maximum eigenvalue, one needs to compute the minimum (over subsets $S\subseteq[n]$) of the maximum eigenvalue of $B'+B'^T$ after projection; such that the maximum eigenvalue of $B+B^T$ is not small, after projection to the coordinates $S$. Two observations are in order:
\begin{enumerate}
    \item Computing the equilibrium of the min-max objective involves computing a minimum over a potentially exponential number of projections corresponding to subsets of $[n]$ -- a non-trivial optimization problem
    \item As $|S|$ increases, the maximum eigenvalue of $B'+B'^T$ will be non-decreasing; if $|S|$ decreases, then the maximum eigenvalue of $B+B^T$ will be non-increasing. Therefore, the optimal value of $|S|$ may be expected to be neither too small nor too large (see Figure~\ref{paths:fig}(c)). 
\end{enumerate}
In Section~\ref{gene4:sec}, we show that a concrete version of the min-max optimization problem is efficiently solvable by using semidefinite programming (SDP) duality to obtain a SDP relaxation of the problem. We then round the vector solution of the SDP relaxation to a $0$-$1$ integer solution which gives the optimal $S$ and the distance between corresponding Lyapunov exponents (see Algorithm~\ref{basket:algo}). In Theorem~\ref{basket:thm}, we show that the objective value computed for the SDP relaxation is close to the equilibrium value in expectation. We illustrate one application of the above algorithm: predicting biomarkers for the design of a subclass of clinical trials -- basket trials. The idea being to identify common biomarkers between different cancers, so that therapies for one can be translated to another (see Section~\ref{gene4:sec} for more details). Currently, such biomarkers are discovered from experiments by trial and error.



\subsection{Genomics contribution}\label{gene5:sec}
In this section, we summarize an application of of our theoretical results above to genomic datasets. Sections~\ref{gene2:sec},~\ref{gene4:sec} and Supplement~\ref{sbs:asump} elaborate on our description here. We use the RNA sequencing data on breast cancer~\cite{metabric},~\cite{metabric2},~\cite{metabric3}, and colorectal cancer~\cite{rc} for our experiments.\footnote{About the datasets: The datasets are available at the cbio portal (\url{https://www.cbioportal.org/}). The breast cancer dataset contains RNA sequencing data (whole genome from tumor biopsy) from about 2,000 cases, such that half  have lymph node involvement (metastasis) and half don't. The colorectal cancer dataset contains RNA sequencing data (whole genome from tumor biopsy) from about 250 colorectal cancer cases, about 60 of which are non-metastatic (grade T2 or below), while the rest have metastasis (grade T3 or above).} Our main contribution, from the perspective of a genomic lens, is an in silico prediction: RORC agonists, which are known to be effective in some colorectal cancers~\cite{rorc-rx,rorc-rx2}, are likely to be effective in some breast cancer cases as well (see also~\cite{rorc-bc-good, rorc-rx-wbc-ct}). More concretely, the steps in the application of our algorithms to obtain the aforementioned conclusion can be summarized as follows:
\begin{enumerate}
    \item Using Algorithm~\ref{algmain}, which assumes the existence of an approximate clock, we reconstruct the noisy linear dynamical system of the underlying gene regulatory network (GRN) for non-metastatic breast and colorectal cancer gene networks (see for e.g., Figures~\ref{dyn:fig}(a) and~\ref{dyn2:fig}(a)), and the corresponding perturbations that lead to metastasis (see for e.g., Figures~\ref{dyn:fig}(b),~\ref{dyn2:fig}(b)).\footnote{We restrict our networks to nuclear receptors for computational limitations, but the analysis can be carried for arbitrary sets of genes with bulk RNA sequencing data.} The recovered networks are not easily interpretable due to their size (they belong to a $2500$ dimensional space for the $50$ odd nuclear receptors that we consider).
    \item We analyze the perturbations (from step (1)), using Algorithm~\ref{basket:algo}, to recover a subset of genes (denoted $\mathcal{C}$), for which the gene expressions for metastatic breast and colorectal cancers will follow similar trajectories, but the evolution of gene expressions of the same set of genes for non-metastatic cases is substantially different for the two cancers (see the discussion in Section~\ref{gene4:sec} for details). For breast and colorectal cancers, some of the genes in $\mathcal{C}$ are: RORC, ESR2, NR2F1 (they correspond to the darker areas in the heatmap in Figure~\ref{paths:fig}(a)).
    \item Next, we compute the genomic interactions (pathways), using Algorithm~\ref{path:alg}, that are predicted to be critical for metastasis for breast and colorectal cancers {\em separately}. From the computed pathways, we select the top five predicted pathways that contain one or more genes in $\mathcal{C}$ (step (2)). Since $\mathcal{C}$ is supposed to identify genes that behave similarly in breast and colorectal cancer metastasis, one expects to find the common critical interactions from this exercise. Indeed, all five of the resulting pathways for colorectal cancer metastasis were mediated via RORC inhibition, and four of the five pathways for breast cancer metastasis were mediated by RORC inhibition  (see Figure~\ref{paths:fig}(c,d)).  Thus, we identified RORC inhibition as a common factor in both cancers, even though its exact role may differ for the two cancers. RORC agonists have been investigated for colorectal and other cancer therapies~\cite{rorc-rx, rorc-rx2}. Our in silico results should motivate further investigation of RORC agonists in breast cancer treatment (see also~\cite{rorc-bc-good, rorc-rx-wbc-ct}).
    \item Finally, in Section~\ref{sbs:asump}, we verify that a weighted average of cyclin expressions\footnote{These are a family of cellular proteins that regulate the cell cycle, so their concentrations peak at the precise points of time in the cell cycle~\cite{cccp}. For example, Cyclin A is active in the S phase, while Cyclin D is active in the transition from the G1 to S phase.} can be used as a sufficiently accurate approximate clock in Algorithm~\ref{algmain}. We use single cell data from~\cite{Bastdas-Ponce} and~\cite{dg}, where the latent time from scVelo~\cite{bergen} acts as our ground truth against which we can compare our approximate clock. We show that the cyclin based approximate clock will recover the same covariance matrix, and hence the linear approximation of the gene network, as the latent time based clock (see also Remark~\ref{crit:rmk} and Theorem~\ref{clock:thm}).
\end{enumerate}


\begin{figure}[htb!]
    \centering
    \begin{subfigure}[b]{0.45\textwidth}
        \centering
        \includegraphics[height=2.3in]{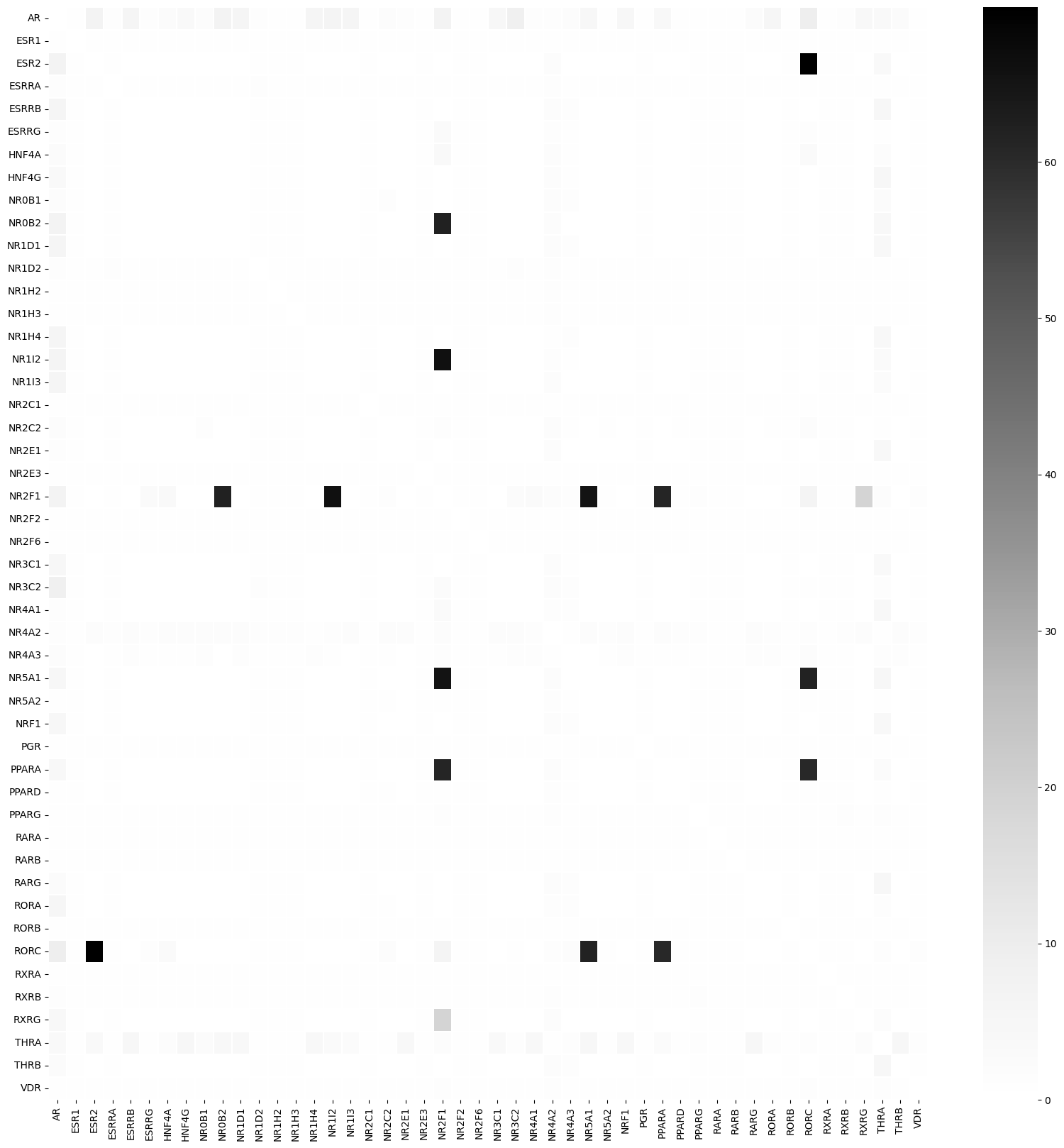}
        \caption{Darker areas identify pairs of genes predicted to have similar expression trajectories in breast and colorectal cancer metastasis. Genes like RORC and NR2F1 could play critical roles in metastasis in both cancers (though their roles may not be the same).}
    \end{subfigure}
    ~
    \begin{subfigure}[b]{0.45\textwidth}
        \centering
        \includegraphics[height=1.5in]{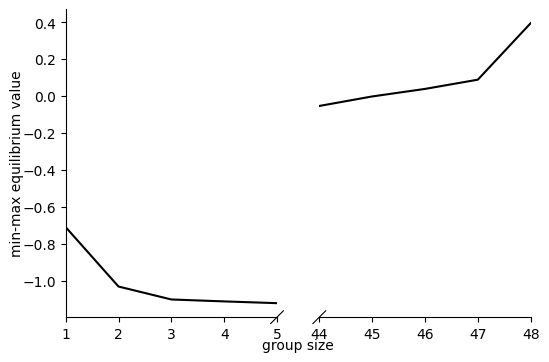}
        \caption{The x-axis shows the dimension $\in[1,50]$, and the y-axis measures the equilibrium value (see Equation~\ref{eqn:obj_lyp}) corresponding to similarity in trajectories between metastatic breast and colorectal cancers relative to their non-metastatic counterparts. Smaller (negative) values imply higher similarity. Thus, there exist local similarities in gene expressions in metastasis, even though globally breast and colorectal cancers have very different gene expression trajectories (as expected).}
    \end{subfigure}
    \centering
    \begin{subfigure}[t]{0.45\textwidth}
        \centering
        \includegraphics[height=1.9in]{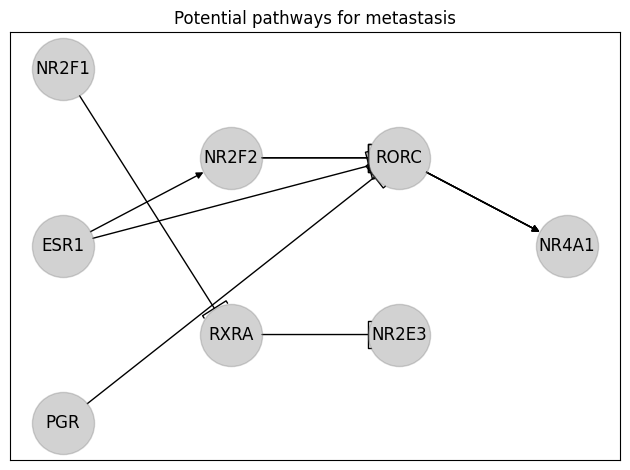}
        \caption{Predicted pathways for metastasis in breast cancer. Note that evidence exists that shows up-regulation of NR4A1 slows disease progression~\cite{nr4a1}.}
    \end{subfigure}%
    ~ 
    \begin{subfigure}[t]{0.45\textwidth}
        \centering
        \includegraphics[height=1.9in]{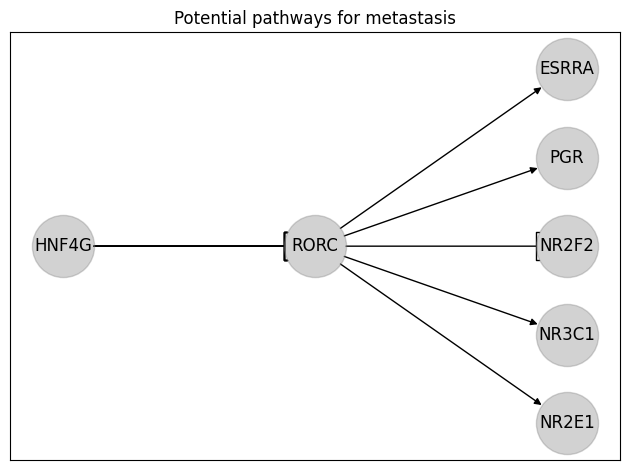}
        \caption{Predicted pathways for metastasis in colorectal cancer. HNF4G expression is known to lead to poor prognosis~\cite{hnf4g}.}
    \end{subfigure}
    \caption{Top: The genes in $\mathcal{C}$ obtained using Algorithm~\ref{basket:algo}. See Section~\ref{gene4:sec} for details. Bottom: Predicted pathways where gene A $\to$ gene B means that A (indirectly) up-regulates B, while {\usebox\inhibit} indicates (indirect) down-regulation. Based on the Bottom Right figure, one potential pathway for metastasis in colorectal cancer could be: HNF4G down-regulates RORC which down-regulates NR2F2. Thus our algorithm suggests that HNF4G eventually up-regulates NR2F2 via a pathway involving RORC. It is known that increased expression of NR2F2 is associated with metastasis in colorectal cancers~\cite{nr2f2}. See Section~\ref{gene2:sec} for details.}\label{paths:fig}
\end{figure}
\noindent Our computations implicitly "average" over inhomogeneous data, and below we discuss resulting limitations.

First, using an approximate clock allows us to overcome the limitation that RNA-seq data is not ordered. However, bulk RNA-seq of tumor specimens contains contributions from many different cell types (cancer cells, stromal cells, immune cells etc). When we reconstruct the approximate GRN (using Algorithm~\ref{algmain}), and use it for predicting critical interactions among genes, we are using the averaged expression data in the different cell types in the tumor, as opposed to any specific cell type. Thus our reconstructed model of the GRN does not capture the GRN within any single cell type, and any conclusions drawn from it are not for any particular cell type. Such "averaged" models are often referred to as mean-field approximations, and we discuss them further in Section~\ref{sbs:asump}.

Second, our in silico results are meant to illustrate the potential of our algorithms . Often the role played by receptors, like RORC, in cancer progression is complex (see for e.g.~\cite{ror-svy}), and using different subcohorts of patients can lead to very different outcomes, due to the inhomogeneous nature of cancer pathways across patients. Using the entire dataset, as we did to obtain Figure~\ref{paths:fig}, can be viewed as "averaging" over many different inhomogeneous pathways from different patient subcohorts. Thus the results obtained may not be true for any given subset of patients. Nevertheless, the fact that we were able to independently predict RORC inhibition as important for disease progression, illustrates that our algorithms may be effective in identifying critical gene targets and pathways, especially if the data is reasonably homogeneous.

\subsection{Related work}
{\em Modeling GRNs:} Understanding transcriptional regulation -- which genes directly regulate a given gene -- equates to the reconstruction of GRNs  from gene expression snapshots. The latter can be measured using RNA-seq data. There is a large volume of work in this area. The surveys~\cite{grn-svy2}, ~\cite{grn-svy} and~\cite{grn-svy3} cover the tools and techniques employed in the study of GRNs. As in our paper, Markov processes have been used to model gene regulation before (see for example~\cite{grn1},~\cite{grn2},~\cite{raj}, and~\cite{elowitz}). The paper~\cite{grn1} gives a detailed overview of modeling transcription regulation by stochastic differential equations (SDEs); and the Ornstein-Uhlenbeck process that we employ in this paper is just a linear approximation of such SDE models that is analytically tractable and yet expressive enough to capture first order behavior. 

The algorithms for reconstructing GRNs fall under two broad classes based on their input: (1) more recent algorithms work with single cell data as input (for e.g.~\cite{scenic}), while (2) older work (for e.g.~\cite{arcane},~\cite{genie}) uses gene expression snapshots as the input. Due to the abundance of unordered RNA-seq data vs single cell data in clinical settings, we have chosen to work with the former. However, the results that use unordered RNA-seq data (like~\cite{cluster},~\cite{graph_model1},~\cite{arcane} and~\cite{genie}) do not formally recover the dynamical system; they recover a graphical approximation of the GRN. The crucial difference between the two is that, while there are various ways to define edge weights in any graphical representation (using mutual information~\cite{arcane}, using feature selection~\cite{genie} etc) and that leads to many different ad-hoc similarity measures between such graphs. On the other hand, there is only one unique linear dynamical system approximation of the underlying Markov process for gene regulation. In this paper, we recover an approximation of that representation under the assumptions in Section~\ref{sbs:asump}. One advantage of recovering the underlying dynamical system is that one can compare similarity in the evolution of trajectories between two such dynamical systems, which corresponds to trajectory of gene expressions. Contrast that with graphical representations, where the walks over such a graph do not necessarily correspond to evolution of trajectories of gene expressions, i.e., the temporal information is not captured.

\noindent {\em Relation to trajectory inference:} Our approximate cyclin based ordering problem is related to the problem of trajectory inference, which has been well explored in {\em single cell} RNA-seq literature~\cite{clk-svy}. Trajectory inference, as opposed to just computing similarity between trajectories via Lyapunov exponents (as in this paper), is likely a harder problem to solve. This is because, we only need to order cell samples well enough to approximately recover a specific covariance matrix in our algorithms, but in the trajectory inference problem one needs to order the cell samples in the correct sequence to predict cellular development. It is possible that cyclin expressions themselves may be effected in many situations by the pathology under consideration (see for e.g.~\cite{fard}). In such cases, our algorithms can not be applied directly with cyclins to construct the approximate clock, but it may be possible to use a different set of gene expressions related to cell division, upstream from cyclins, that are verified to be not effected by the pathology, to construct the approximate clock, on a case by case basis.\\
{\em Connections to basket trials:} A basket trial is a clinical trial aims to group patients by their shared biomarker, as opposed to a disease, and gives them the same treatment that targets the shared biomarker (for e.g.~\cite{nci-match}). Such trials have been critical in discovering novel therapies like anti-PD1 therapies (for e.g.~\cite{pd1}). Algorithm~\ref{basket:algo} provides a way to predict biomarkers for basket trials. For example, in Figure~\ref{paths:fig}(c,d), we predict that RORC inhibition is common to breast and colorectal cancer metastasis, and RORC agonists have been investigated for colorectal cancers~\cite{rorc-rx2}. This raises the prospect of investigating RORC agonists for breast cancer as well (see also~\cite{rorc-bc-good}). We believe our approach can be useful for predicting biomarker targets across different cancers. 
\noindent {\em Relation to filtering:} Perhaps the line of work closest to Algorithm~\ref{algmain} is Kalman filtering in the presence of network loss and delays (see~\cite{yang},~\cite{liu} and~\cite{nikfe}, and the references therein). However,  such results  typically assume a independent identical probability of packet drop or delay, sometimes modeled by a Markov process of interest. Also, although the idea of computing conditioned Gaussian distributions (for the error upper-bound in Theorem~\ref{clock:thm}) is similar to that used for the analysis of the Kalman filter, the problem studied in our paper differs from the standard filtering problem because we are interested in recovering the underlying linear system and not in predicting the next step of the system.\\
\noindent {\em Relation to min-max principles:} Perhaps the problem closest to Algorithm~\ref{basket:algo} (and Equation~\ref{eqn:obj_lyp}) are versions of the well-known Courant-Fisher min-max principle, which involves a direct computation of the eigenvalues of symmetric matrices by repeatedly computing orthogonal subspaces. However, in our problem, we can not choose the orthogonal subspaces iteratively -- the basis is fixed, and we need to compute a single subspace corresponding to a subset of the $[n]$ variables in Equation~\ref{eqn:obj_lyp}.


\section{Technical contributions}\label{ovw:sec}
In this section, we provide the abstract algorithms (Subsections~\ref{th1:sbs} and~\ref{th2:sbs}), and elaborate on their genomics applications (summarized in Subsection~\ref{gene5:sec}): step (1) in Subsection~\ref{gene2:sec}, while steps (2) and (3) in Subsection~\ref{gene4:sec}. The supplement contains the proofs (Sections~\ref{thm:sec} and~\ref{thm2:sec}), and empirical verification of assumptions (Sections~\ref{sec:synth} and~\ref{app:sec}).
\begin{definition}
Suppose we have a $n$-dimensional Ornstein-Uhlenbeck (OU) process 
\begin{equation}\label{oup:dfn}
    dx_t=Bx_tdt+AdW_t,
\end{equation}
where $W_t$ is $n$-dimensional standard Brownian motion, $A$ is the {\em diffusivity} matrix, $AA^T$ is the {\em covariance matrix of increments}, and $B$ the {\em drift} matrix. (see for example~\cite{oks})
\end{definition}
Suppose that we observe $x_t$ at various time points in $[0,T]$ for large $T$, but we either do not have the ordering information or do not keep the samples ordered. Thus the covariance matrix of the increments of $x_t$ is not assumed to be known in this paper, since it requires ordering information for its computation. Then, we want to know whether we can recover the underlying (unknown) $n\times n$ drift matrix $B$, that characterizes our dynamical system.\footnote{We have assumed the system is centered about $0$ but that is not an issue since the long term average can be computed without knowing the ordering and subtracted from the observations.}

\begin{remark}\label{crit:rmk}
If we are given an ordering of $x_{t}$s (correct or with errors) as $(x_{t_0},x_{t_1},...,x_{t_m})$ then we may compute covariance matrix of increments, with respect to that ordering, as the average $\frac{1}{m}\sum_{k=1}^m(x_{t_k} - x_{t_{k-1}})^T(x_{t_k} - x_{t_{k-1}})$. If $t_0\le t_1..\le t_m$ then $AA^T$ is the empirical covariance matrix of increments and is closely approximated by the average, by Sanov type concentration theorems (See~\cite{dz}). Moreover, being an average, this matrix is fairly robust to noise and small errors; a property that is critical for the algorithms in this paper.
\end{remark}

In general, reconstructing $B$ without ordering information for the $x_t$ may not be possible. However, if: (1) there exists an "approximate clock" (see Definition~\ref{appx-clk:dfn}) and (2) $B$ is row sparse, i.e., most of the entries in each row of $B$ are zero or close to zero; then we can approximately reconstruct $B$ with few samples in polynomial time, by solving a convex optimization problem (Algorithm~\ref{algmain}). 
\begin{definition}\label{appx-clk:dfn}
Given an OU process as in Equation~\ref{oup:dfn}, we say that there exists an {\em approximate clock} if we know that some non-empty subset 
of the $n$ coordinates (of $x_t$), or some projection of the process, has a positive drift in each coordinate, and the trace of the diffusivity matrix of the process is smaller than all the drift coordinates.\footnote{The idea being that the clock process is likely to be increasing for most finite length time intervals.}
\end{definition} 
To keep the presentation simple, we will assume that if there exists a single coordinate of $x(t)$,\footnote{Here we assume a single coordinate exists as opposed to multiple coordinates or allowing a linear combination of coordinates} denoted by $\tau(t)$, and it is modeled by the diffusion:
\begin{equation}
    d\tau(t)=\delta dt+\eps dW'_t,
    \label{eq:diffusion}
\end{equation}
where $W'_t$ is standard Brownian motion and $\delta$ is positive and much greater than $\eps^2$; then $x(t)$ is said to admit an approximate clock. We use $x_t$ and $x(t)$ interchangeably.

Assuming $B$ is row sparse, i.e., each row of $B$ has only $s=\log^{O(1)}(n)$ non-zero entries, and the number of samples $t$ is larger than $s\log(n)$, and the existence of an approximate clock; our {\em main algorithmic contribution} is Algorithm~\ref{algmain}, that can approximately reconstruct $B$ from the $t$ samples. 

We will use the approximate clock to recover an approximation of the matrix $A$, which we denote by $\tilde{A}$. We won't actually need the ordering from approximate clock except to compute $\tilde{A}$, which gives the approximate covariance of the increments: $\tilde{A}\tilde{A}^T$. In Theorem~\ref{clock:thm} (stated with proof in the supplement), we formally show that ordering the $x_i$ values using $\tau(t)$ and computing the covariance from the increments ($\tilde{A}\tilde{A}^T$) closely approximates the covariance of the increments ($AA^T$) computed using the ordering based on the actual time.

\subsection{Algorithm to reconstruct dynamical system}\label{th1:sbs}
Algorithm~\ref{algmain} (below) reconstructs a dynamical system given an approximate clock. It runs in polynomial time, since the main steps are computing a covariance matrix and solving the convex optimization problem in Step~\ref{step:alg1step5} of Algorithm~\ref{algmain}, which is a second order cone program, and a slight variation of basis pursuit in the compressed sensing literature (see for example~\cite{rauhut}). Therefore unique reconstruction is guaranteed under sparsity assumptions. In particular the following uniqueness theorem holds true for the matrix $\tilde{B}$ recovered by Algorithm~\ref{algmain}. 
\begin{theorem}\cite{dirk}\label{uniqueness:thm}
Suppose each row of $B$ has at most $s$ non-zero coordinates, 
then for $m=\Omega(\log n)$ large and the sample times $\{t_1,...,t_m\}$ for samples $\{x_{t_1},...,x_{t_m}\}$ are uniformly distributed,
then with probability $1-o(1)$ (for large $n$): $\|\tilde{B}-B\|_1=0$, i.e., $B$ can be recovered uniquely by solving a convex program. If, however, the rows are not exactly $s$ sparse and the $\ell_1$ norm of the $n-s$ remaining entries in any row is at most $\epsilon$, then we have: $\|\tilde{B}-B\|_1\le \epsilon n$.
\end{theorem}
While the proof follows by easy modifications of known results in existing compressed sensing literature (see~\cite{rauhut} or~\cite{dirk}), we note the one novelty here over existing compressed sensing literature: the formulation of the constraint in Step~\ref{step:alg1step5} of Algorithm~\ref{algmain} uses the characterization of the long term variance of the OU process (see~\cite{oks}). 

Algorithm~\ref{algmain} reconstructs an approximation of the drift matrix $B$ from the covariance matrix $\Sigma$ and the (approximate) covariance matrix of the increments $\tilde{A}\tilde{A}^T$ up to an additive factor. It is sufficient to use $\tilde{A}\tilde{A}^T$ instead of $AA^T$ as is justified by Theorems~\ref{uniqueness:thm} and~\ref{clock:thm}.

\begin{algorithm}[htb]
\caption{Recover dynamical system}\label{algmain}
\begin{algorithmic}[1]
\State {\bf Input:} Unordered snapshots $\{x_1,..., x_t \in\RR^n; t\in[0,T]\}$ derived from an Ornstein-Uhlenbeck process $dx_t=Bx_tdt+AdW_t$; where the diffusivity matrix $A$ is known up to a constant additive error as $\tilde{A}$, and $B$ is row sparse but unknown.
\State {\bf Output:} A recovered matrix $\tilde{B}$ that is close (in $\ell_1$ norm) to $B$.
\Statex $\triangleright$ {\bf Algorithm starts:}
\State Compute the $n\times n$ covariance matrix of the random vectors $\{x_1,...,x_t\}$, denoted $\Sigma$ 
\Statex $\triangleright$ Note:(1) computing $\Sigma$ does not require ordering the $x_i$s 
\Statex $\triangleright$ $\qquad$ (2) Using $\tilde{A}$ instead of $A$ suffices (see Theorems~\ref{uniqueness:thm} and~\ref{clock:thm})
\Statex $\triangleright$ $\qquad$ (3) $B(i,\cdot)$ denotes the $i^{th}$ row of $B$
\State Solve $\tilde{B}:=\arg\min_B \sum_{i=0}^{n}\|B(i, \cdot)\|_2$$~\quad$ s.t. $\Sigma B+B^T\Sigma=-\frac{\tilde{A}\tilde{A}^T}{2}$\label{step:alg1step5}
\State return $\tilde{B}$
\end{algorithmic}
\end{algorithm}
\begin{remark}
Note that the only use of ordering using approximate clock is to recover the matrix $\tilde{A}\tilde{A}^T$ in Algorithm~\ref{algmain}, we do not use the  ordering information anywhere else. 
\end{remark}

Often it is not enough to recover the drift matrix, but we need to distinguish one linear system from another, where the first is a small perturbation of the second, and one has only few samples of the second at hand. This exact case happens for genomic data samples, where one system is the reference and many samples are available, but another system is some rare genetic condition or disease for which fewer samples are available. Algorithm~\ref{algpert} obtains the {\em perturbation matrix} $P$ (within a small error) that distinguishes the second linear system from the first. Note that the algorithm always returns a value for a large enough choice of the noise parameter $\eta$ and small enough choice of $\epsilon$ in Step~\ref{step:alg1step8} of Algorithm~\ref{algpert}, since that will suffice to make the convex program feasible. Hence we chose a small enough $\eta$ (via binary search) that just makes the convex program feasible. 


\begin{algorithm}[htb]
\caption{Order and recover perturbations between two dynamical systems}\label{algpert}
\begin{algorithmic}[1]
\State {\bf Input:} Unordered snapshots $\{x_1,..., x_t \in\RR^n; t\in[0,T]\}$ and $\{x'_1,...,x'_{t'}\in\RR^n; t'\in[0,T']\}$ derived from two Ornstein-Uhlenbeck processes $dx_t=Bx_tdt+AdW_t$ and $dx'_t=B'x_t'dt+A'dW_t$; where the diffusivity matrices $A$ and $A'$ are known up to small error, but $B$ and $B'$ are unknown.
Moreover, $t'\ll t$ and $B$ and $B'$ are row sparse. 
\State {\bf Output:} A recovered matrix $\tilde{B}$ that is close to $B$ in $\ell_1$ norm, and a perturbation matrix $\tilde{P}$ such that $\tilde{B}+\epsilon\tilde{P}\simeq B'$, for a small positive parameter $\eps$.
\Statex $\triangleright$ {\bf Algorithm starts:}
\State Use Algorithm~\ref{algmain} to compute $\tilde{B}$
\Statex $\triangleright$ The smaller number of samples prevents direct computation of $\tilde{B'}$, so we compute a first order approximation.
\State Compute the $n\times n$ covariance matrix of the random vectors $\{x'_1,...,x'_{t'}\}$, denoted $\Sigma'$
\Statex $\triangleright$ Choose a small noise $\eta$ so that the following convex program program is feasible
\Statex $\triangleright$ Note: $P_{ij}$ denotes the $i^{th}$ row, $j^{th}$ column element of matrix $P$
\State Solve $\tilde{P}:=\arg\min_P \sum_{i,j=0}^n |P_{ij}|$$~\quad$s.t. $\|(\Sigma-\Sigma')\tilde{B}^T+\tilde{B}(\Sigma-\Sigma')-\epsilon(\Sigma'P^T+P\Sigma')\|_2\le\eta$\label{step:alg1step8}
\State return $\tilde{P}$
\end{algorithmic}
\end{algorithm}

\subsection{Algorithm for predicting genomic pathways}\label{gene2:sec}

As a concrete application, we use the recovered perturbation matrix $\tilde{P}$ to isolate the paths consisting entirely of high weight edges in the directed graph (defined below) underlying the dynamical system, since these paths should reflect the prominent genomic pathways that differentiate the reference dynamical system from the pathological dynamic system. We note that more standard methods like time series analysis based on Fourier transform of the recovered linear system can also be tried out, but our method below is far simpler, and thus probably more robust to recovery errors, and it suffices to illustrate the pathways that can be recovered.

The matrix $P$ can be thought of as a directed graph with entry $P_{ij}$ reflecting a weighted directed edge from the node corresponding to gene $j$ to that for gene $i$. A {\em prominent} path is defined to be one such that all its edges have weight higher than some fixed threshold, say $\theta$. For such a prominent path $p$ to potentially reflect an actual underlying genetic pathway that promotes the pathological phenotype, one of the following intuitive conditions should hold:

\noindent 1. If the gene corresponding to the terminal node of $p$ positively correlates with the pathological phenotype then it should be up-regulated by the gene preceding it in the path $p$; moreover, the same property or property (2) should  hold for the subpath terminating in the penultimate gene, and so on.

\noindent 2. Otherwise, if the gene corresponding to the terminal node of $p$ negatively correlates with the pathological phenotype then it should be down-regulated by the gene preceding it in the path $p$; moreover, the same property or property (1) should now hold for the subpath terminating in the penultimate gene, and so on.

This intuition leads to Algorithm~\ref{path:alg} which predicts prominent genomic pathways that may be experimentally verified to confirm that some or all of them lead to the pathological phenotype.

\begin{algorithm}[htb]
\caption{Recover Pathways}\label{path:alg}
\begin{algorithmic}[1]
\State {\bf Input:} The underlying linear dynamical systems matrices $\tilde{B}$ and $\tilde{B'}$ and the correlations between each of the coordinates (corresponding to genes) and the two phenotypes of interest.
\State {\bf Output:} A set of pathways of a given length $L$ that are prominently different between the two phenotypes.
\Statex  $\triangleright$ {\bf Algorithm starts:}
\State Compute and sort the list of genes in descending order of the absolute value of their correlation coefficient with the pathological phenotype, denote this list as $\vec{g}$
\State Compute $C:= \mathrm{Diag}(\vec{g})(\tilde{B'}-\tilde{B})$
\State Fix a positive threshold $\theta$ and set $C_{ij}=0$ if $C_{ij}<\theta$, denote the resulting matrix as $\Pi_\theta$ 
\State Compute the set of paths of length $L$ in the graph with adjacency matrix $\Pi_\theta$ and return them.
\end{algorithmic}
\end{algorithm}

\noindent Finally, we explain the steps required to generate pathways in Figure~\ref{paths:fig}(c,d) using Algorithms~\ref{algmain},~\ref{algpert} and~\ref{path:alg} and RNA-seq data in the breast cancer dataset~\cite{metabric,metabric2, metabric3} (the colorectal cancer dataset~\cite{rc} is similar). 
\begin{enumerate}
    \item First, we use Algorithm~\ref{algmain} to compute the underlying linear dynamical system matrix: $\tilde{B}$ (Figure~\ref{dyn:fig}(a) and Figure~\ref{dyn2:fig}(a))
    \item Second, we recover the perturbations $\tilde{P}$ from $\tilde{B}$ using Algorithm~\ref{algpert} (Figure~\ref{dyn:fig}(b)and Figure~\ref{dyn2:fig}(b))
    \item Finally, we use Algorithm~\ref{path:alg} with $\theta=0.05$ to reconstruct the pathways similar to those in Figure~\ref{paths:fig}(c,d).
\end{enumerate}
Since the dataset samples are unordered we use a linear combination of Cyclin A and D expression levels, as our approximate clock, to order the data-points by cell age. Based on known literature about human cyclins, if we were to obtain a random sample of two cells from a tissue and found the difference, i.e., Cyclin A - Cyclin D levels, was relatively higher for sample 1 over sample 2 than we know that sample 1 was more advanced in the cell cycle than sample 2 (perhaps with occasional error). Thus we use the difference in their expression levels as our approximate clock. Note that we can't verify how good this clock is, since we don't have access to the ground truth that orders cells by their relative cell cycle stage (but we do this verification in single cell datasets in Section~\ref{sbs:asump}). 




\subsection{Algorithm for differentiating two dynamical systems}\label{th2:sbs}
Suppose we are given two linear dynamical systems in $\RR^n$ that start from the point $x_0$:
\begin{eqnarray}
    \dot{x} &=& Bx,\\
    \dot{x} &=& B'x.
\end{eqnarray}
It is well known that, for linear dynamical systems like $\dot{x}=Bx$, whether the long term trajectory converges to an equilibrium point or not is determined by the real part of the maximum eigenvalue of $B$ (the Lyapunov exponent). Thus, if the maximum eigenvalue $\lambda_B$ of $\frac{B+B^T}{2}$ is negative then $x(t)$ will converge to an equilibrium point for large $t$, otherwise it grows (diverges) with rate $\exp(\lambda_B t)$. We want to answer the following question: does there exists a subset of coordinates $S^*\subset[n]$ such that the long term trajectory of $B$ diverges from any equilibrium point as $\gtrsim\exp(\delta t)$ for some positive parameter $\delta$, while the long term trajectory of $B'$ diverges from any equilibrium point as $\lesssim\exp(\delta' t)$ for some positive parameter $\delta'$, such that $\delta-\delta'\gg 0$ ?\footnote{The difference of $\delta$ and $\delta'$ is computationally tractable as opposed to their ratio.} For a fixed choice of scaling parameter $\alpha$, the answer can be written as the following optimization problem:
\begin{equation}
    \mathsf{O}=\min_{S\subseteq[n]}\left(\max_{x\in\RR^n,\atop{\|x\|_2\le 1}} x^T\mathrm{Diag}(S)\cdot \frac{B'+B'^T}{2}\cdot\mathrm{Diag}(S)x -\alpha\max_{z\in\RR^n,\atop{\|z\|_2\le 1}} z^T\mathrm{Diag}(S)\cdot \frac{B+B^T}{2}\cdot\mathrm{Diag}(S)z\right),\label{eqn:obj_lyp}
\end{equation}
where $\mathrm{Diag}(S)$ denotes the $n\times n$ diagonal matrix with $\{0,1\}$ diagonal entries that are $1$ at position $(i,i)$, if and only if $i\in S\subseteq[n]$. If $\mathsf{O}\gg 0$ then we have answered the question above in the positive, and vice versa. The above optimization problem is non-convex as it requires $\{0,1\}$ solutions for $\mathrm{Diag}(S)$. It is computationally tractable by brute force only for small $n$ and small optimal $|S|$, which suffices for our purposes (see Subsection~\ref{gene4:sec}). But, more generally, for large $n$, observe that its semidefinite relaxation below, which replaces vectors with positive semidefinite matrices, is concave-convex.
\begin{eqnarray}
    \mathsf{\tilde{O}}=\min_{Y}&\max_{\atop{X, Z}}\left( \langle \frac{B'+B'^T}{2}, X\odot Y\rangle - \alpha\cdot \langle \frac{B+B^T}{2}, Z\odot Y\rangle\right)\label{eqn:obj_lyp2}\\
    \mathrm{\ }\quad & \forall i\in[n]: Y_{ii}\le 1,\label{eqn:constr_lyp3}\\
    \mathrm{\ }\quad & \mathrm{Tr}(X)\le 1,~\mathrm{Tr}(Z)\le 1,\label{eqn:constr_lyp4}\\
    \mathrm{\ }\quad & X,Y,Z\succeq 0\label{eqn:constr_lyp5},
\end{eqnarray}
where $\odot$ denotes the Hadamard product (entry-wise  multiplication) of two matrices, and $\langle, \rangle$ denotes the inner product between two matrices (treated as vectors). Observe that the objective in Equation~\ref{eqn:obj_lyp2} is linear in $X,Z$ for a fixed $Y$, and the constraints in Equations~\ref{eqn:constr_lyp3},~\ref{eqn:constr_lyp4} and~\ref{eqn:constr_lyp5} are convex. 

Although, there exists an equilibrium solution to the optimization problem above, it is not immediately computable from the above formulation. However, using SDP duality allows us to convert the min-max optimization into a SDP with just a maximization objective. The resulting SDP relaxation is given below:
\begin{eqnarray}
    &\ &\sf{\bar{O}}=\min_{u,v\in\RR,\atop{Y\in\RR^{n\times n}}} \left(u+v\right)\label{eqn:obj_lyp3}\\
    \mathrm{\ }\quad & \mathrm{Diag}((u,u,...,u)) &\succeq \frac{B'+B'^T}{2}\odot Y,\label{eqn:constr_lyp6}\\
    \mathrm{\ }\quad & \mathrm{Diag}((v,v,...,v)) &\succeq -\alpha\cdot\frac{B+B^T}{2}\odot Y,\label{eqn:constr_lyp7}\\
    \mathrm{\ }&\quad & \forall i\in[n]: Y_{ii}\le 1,\label{eqn:constr_lyp8}\\
    \mathrm{\ }&\quad & Y\succeq 0\label{eqn:constr_lyp9},
\end{eqnarray}
where $\mathrm{Diag}$ denotes the $n \times n$ diagonal matrix. The SDP solution is eventually rounded to the $0$-$1$ solution, and the details are provided in Algorithm~\ref{basket:algo}.
\begin{algorithm}[htb!]
\caption{Compute distance between Lyapunov exponents}\label{basket:algo}
\begin{algorithmic}[1]
\State {\bf Input:} The $n\times n$ linear dynamical systems interaction matrices $\tilde{B}$ and $\tilde{B'}$, and parameter $\alpha>0$.
\State {\bf Output:} The coordinates $S$ for the approximate minimal Lyapunov exponent from Equations~\ref{eqn:obj_lyp} (see Theorem~\ref{basket:thm}).
\Statex $\triangleright$ {\bf Algorithm starts:}
\State Solve the SDP relaxation in Equations~\ref{eqn:obj_lyp3},~\ref{eqn:constr_lyp6},~\ref{eqn:constr_lyp7} and~\ref{eqn:constr_lyp8} to obtain $X^*, Y^*$ and $Z^*$.
\Statex $\triangleright$ For a positive semidefinite matrix $Y$ with diagonal entries $Y_{ii}\in[0,1]$, let $\mathrm{Ber}(Y)$ denote the $n$-dimensional $0$-$1$ 
\Statex $\triangleright$ Bernoulli random variable with first moment as the diagonal elements $Y_{ii}$ and second moment matrix set as $Y$.
\State For $R\in\{0,1\}^{n}$, choose $R\sim\mathrm{Ber}(Y^*)$
\State Set $S$ using $R$ as the indicator vector to return: (1) $S$, and (2) the objective value in Equation~\ref{eqn:obj_lyp3} using the $0$-$1$ solution, i.e., $R$.
\end{algorithmic}
\end{algorithm}
The distance computed by Algorithm~\ref{basket:algo} is equal to the equilibrium solution in expectation (over $R$), and is proven in Theorem~\ref{basket:thm} (stated with proof in the supplement). Note that algorithm runs in polynomial time as it solves a convex optimization problem.

\subsection{Identifying genes for basket trials}\label{gene4:sec}
Suppose we have reconstructed the linear approximation of the breast cancer and colorectal cancer GRNs using only non-metastatic portion of the sequencing data and the algorithm in Subsection~\ref{gene2:sec}. If we focus only on their mean trajectories, then we can ignore the noise term (which has mean zero, so it vanishes in expectation). As in Subsection~\ref{gene2:sec}, let the interaction matrix of the two GRNs be denoted as $B_{bc}$ and $B_{rc}$, respectively. Furthermore, let $\Delta B_{bc}$ and $\Delta B_{rc}$ denote the perturbation matrices that represent a shift towards metastasis in the corresponding GRNs. 

Relative to the trajectory of metastatic colorectal cancer GRN, the trajectory of metastatic breast cancer GRN evolves as $\dot{x} = ((B_{bc}+\Delta B_{bc})-(B_{rc}+\Delta B_{rc}))x$. Similarly, for the non-metastatic case this relative trajectory evolves as: $\dot{y} = (B_{bc}-B_{rc})y$. Suppose we want to compute a set of coordinates (genes) $S\subseteq[n]$, such that the projection of $x(t)$ to $S$ (denoted $\pi_S(x(t))$) converges but $\pi_S(y(t))$ diverges. In other words, the GRNs for the two cancers evolve similarly for metastatic cases, but the  GRNs for the two cancers evolve differently for non-metastatic cases. If instead of "convergence", we simply requiring the divergence rate for $\pi_S(x(t))$ be smaller than that of $\pi_S(y(t))$, then the problem of computing $S$ reduces to that in Subsection~\ref{th2:sbs}. Thus, sets $S$ for which Equation~\ref{eqn:obj_lyp} has significantly negative objective value, contain the genes that are common to metastasis across the two cancers. Such "useful" sets are neither too small nor too large, as argued below.

For a small set of coordinates (genes) $S$, say $|S|=1$, the Lyapunov exponents of $\pi_S(x(t))$ and $\pi_S(y(t))$ will have similar values, making their difference less negative (see Figure~\ref{paths:fig}(b)), as there will often be some gene for which both GRNs have large Lyapunov exponents. Hence, focusing too much on one gene (or too small a number of genes) makes it harder to find similarity in evolution of trajectories in metastasis across cancers. On the other hand, for $|S|=n$ (all genes), the Lyapunov exponents for $x(t)$ and $y(t)$ will again be large, as the two GRNs will have different evolution overall. Thus the objective in Equation~\ref{eqn:obj_lyp} should be relatively large for small and large $|S|$. Such $S$ are unlikely to correspond to common evolution of metastasis. However, for intermediate $|S|$, the GRNs for metastatic cases evolve similarly, but the GRNs for non-metastatic cases evolve differently. {\em Figure~\ref{paths:fig}(b) plots the objective (Equation~\ref{eqn:obj_lyp}), for different values of $|S|$ (the group-size)}.

Therefore, in order to obtain the critical genes $\mathcal{C}$ (see Subsection~\ref{gene5:sec}), i.e., Figure~\ref{paths:fig}(a), we use projections over subsets of genes of size $2$ or more. We enumerate all subsets of genes $S$ with $|S|\le 5$,\footnote{The bound $5$ is picked somewhat arbitrarily, but it's already $\frac{n}{10}$ for our dataset with $n\simeq50$.} and evaluate the ratio of the Lypanunov exponent  of non-metastatic GRN to the metastatic one, for each $S$ (denote it $r_S$).\footnote{Recall that, for any given $S\subseteq[n]$, higher the maximum eigenvalue restricted to the sub-matrix corresponding to $S$, more dissimilar are the trajectories $\pi_S(x(t))$ and $\pi_S(y(t))$, and vice versa. The ratio will be non-negative as long as the submatrices are not negative definite (unlikely in practice). We also exclude outliers, i.e., any $S$ for which the denominator of $r_S$ is too small.} If genes $i,j\in S$ then we add weight $r_S$ to $(i,j)$. The total weight of a pair of genes $(i,j)$ is defined as: $r((i,j)):=\sum_{\{i,j\}\in S}r_S$.
Thus, $r((i,j))$ is higher, if the expression of the two genes $i,j$ evolves similarly in metastasis than in non-metastasis. {\em We plot $r((i,j))$ for all pairs $i,j$ of nuclear receptors, in Figure~\ref{paths:fig}(a)}. If some square corresponding to genes $S=\{g_1,g_2\}$ is dark in Figure~\ref{paths:fig}(a), those gene expressions should evolve similarly in metastatic breast and colorectal cancers, but they should evolve differently in non-metastatic breast and colorectal cancers. RORC and ESR2 denote one such gene pair. 
\newpage

%
%
%


\bibliographystyle{plain}  
\bibliography{main}

\newpage
\appendix
\onecolumn
\section{Underlying assumptions and their verification}\label{sbs:asump}

Our algorithms, specifically Algorithm~\ref{algmain}, rely on three assumptions: (1) sparsity: the rows of the matrix,\footnote{This is the matrix that constitutes of the coefficients of the first order (degree $1$) terms in the linearization of the underlying dynamical system.} are zero (or near zero) in most entries, (2) a mean field model assumption: the expressions in the RNA-seq data-set represent an ``average" cell in the tissue concerned, and (3) existence of an approximate clock: we should be able to (roughly) order our cell samples according to cell age (measured as time since last cell division in our paper, but it can be replaced with any appropriate temporal ordering based on the application at hand).
We now justify these assumptions.

First, the sparsity assumption is expected to hold for GRNs because the transcriptional regulation of a given gene directly relies on only a handful of polymerases, promoters and previous gene products (see, for example~\cite{grn1}). Even if we replace direct regulation by downstream regulation within a few steps, for example when studying GRNs restricted to nuclear receptor genes only, we still expect the resulting matrix to be sparse.

Second, we assume that the population weighted mean expression of any given gene in the tissue sample, that is reflected in the bulk RNA-seq expression value for that gene, is a good approximation of the gene expression in the tissue. Such an assumption is made in mean field theories. The mean field model is usually a valid approximation of a dynamical system when: (a) the sample size is large and (b) fluctuations in value are small (cf. the Ginzburg criterion). That fluctuations should be small, i.e., the variance for a given gene expression should be small, is intuitive, because biological systems are usually in equilibrium (homeostasis). For an organism to be able to maintain equilibrium, large fluctuations in gene expression should be a rare occurrence, as successfully reverting the system back near its mean is harder to do from distant outlier values. So if the sample size is large enough, then the mean field assumption should be valid. 

Finally, before justifying the approximate clock assumption, we first discuss the reason we assume the existence of an approximate clock. We work with RNA-seq data, which is more readily available, unlike single-cell RNA-seq data. In single-cell RNA-seq data, it is possible to order cells by their latent time (see for example~\cite{hku}).  However, for RNA-seq data, latent time is not available, as latent time is a model-based construct. One saving grace is that {\em our algorithms do not use the cell sample ordering except to reconstruct the covariance matrix of increments, which is a pretty robust to noise}, in that it is  only significantly effected if the approximate ordering used has too many errors. Therefore, using appropriate cyclin expression levels to order cell samples in readily available RNA-seq data, although noisy, is 
sufficiently accurate for our purposes; as evidenced by our verification on some single cell RNA-seq datasets, where we assume latent time as the ground truth (see Figures~\ref{scpancreas:fig} and~\ref{dg-fig}). 

To elaborate, we verify that the reconstructed covariance matrix of increments is indeed robust to a small number of errors in ordering that may arise from using an {\em approximate} clock in three different ways: (1) theoretical justification: we prove that the reconstructed covariance matrix is correct 
(Theorem~\ref{clock:thm} in Subsection~\ref{thm:sec}), (2) by simulation on synthetic data in Subsection~\ref{sec:synth},
and (3) by comparing the approximate ordering based on cyclin expressions with the ordering based on cell latent time in murine pancreatic single cell RNA-seq dataset from~\cite{Bastdas-Ponce} (Figure~\ref{scpancreas:fig}) and murine dentate gyrus dataset of~\cite{dg} (Figure~\ref{dg-fig}) in Subsection~\ref{app:sec}.\footnote{The python code for empirical results is available at: https://github.com/cliffstein/recomb\_pathways.} In each case, we find that the covariance matrix computed using the approximate clock will be close to the actual covariance matrix.

\subsection{Theoretical bounds on error from an approximate clock}\label{thm:sec}

In this section, we derive an upper-bound on the error of the recovered diffusivity matrix when we use an approximate clock. We will assume that our dynamical system can be modeled by an OU process and it has a coordinate, or a linear combination of coordinates, that admit a positive drift which is larger than the diffusivity. This coordinate will act as an approximate clock. Theorem~\ref{clock:thm} shows that 
if we use this approximate clock then the recovered diffusivity is not too different from the diffusivity that would be obtained if we knew the ordering in the dynamical system. 

\begin{theorem}\label{clock:thm}
Assuming that the clock coordinate $\tau(t)$ is independent of the remaining coordinates of $x_t$, the diffusivity matrix $\tilde{A}\tilde{A}^T$ computed using the approximate clock $\tau(t)$ is asymptotically equal to the actual diffusivity, i.e.,
for $\eps\ll \delta$,\footnote{(1) Since $\E[\tau(t)]=\delta t$, the normalization of $\tilde{A}\tilde{A}^T$,  has a $\delta T$ term as opposed to a $T$ term, (2) one can skirt around the issue of negative values of $\tau$ by assuming that $\tau(0)$ is large,  ensuring $\tau$ remains positive with high probability.}
$$
    \tilde{A}\tilde{A}^T
     =  \lim_{h\to 0}\E\left[\frac{1}{\delta T}\int_0^T\langle (x(\tau(t+h))-x(\tau(t))), (x(\tau(t+h))-x(\tau(t)))\rangle d\tau(t)\right]\simeq AA^T\ .
    \label{eq:approx-clock}
$$

\end{theorem}
\begin{proof}
Note that
\begin{equation}
    \tilde{A}\tilde{A}^T=\lim_{h\to 0}\E\left[\frac{1}{\delta T}\int_0^T\langle (x(\tau(t+h))-x(\tau(t))),(x(\tau(t+h))-x(\tau(t)))\rangle d\tau(t)\right]
\end{equation}
by definition of $x_t$.
By the tower property of conditional expectation and switching the order of integrals:
\begin{equation}
    \E\left[\frac{1}{T}\int_0^T\langle dx(\tau(t)),dx(\tau(t))\rangle\right]=\E\left[\frac{1}{T}\int_0^T\E\left[\langle dx(\tau(t)),dx(\tau(t))\rangle\biggr\rvert \tau(t)\right]\right]
\end{equation}
Recall that,
\begin{eqnarray}
    dx(t)&=& Bx(t)dt+AdW_t\\
    d\tau(t)&=& \delta dt+\eps dW'_t,
\end{eqnarray}
where $W'_t$ and $W_t$ are independent standard Brownian motions in $\RR$ and $\RR^n$ respectively.
The inner conditional expectation can be evaluated in terms of $\tau(t)$ as follows:
\begin{equation}
\E\left[\langle dx(\tau(t)),dx(\tau(t))\rangle\biggr\rvert \tau(t)\right]=\E\left[\langle Bx(\tau)d\tau+AdW_\tau, Bx(\tau)d\tau+AdW_\tau\rangle\biggr\rvert \tau(t)\right].
\end{equation}
Note that we are conditioning a Gaussian with another Gaussian, so the result is a Gaussian variable.
The quadratic variation term in the RHS expectation is a sum of three types of terms that can be calculated using:
\begin{eqnarray}
    \E\left[\langle AdW_\tau, AdW_\tau\rangle\right]&=& AA^Td\tau(t)\\
    \E\left[\langle Bx(\tau)d\tau, Bx(\tau)d\tau\rangle\right]&=&\eps^2Bx(\tau)x(\tau)^TB^Tdt\\
    \E\left[\langle AdW_\tau, Bx(\tau)d\tau\rangle\right]&=& 0.\label{cov:eqn}
\end{eqnarray}
The last equality follows because we have assumed that $W'_t$ and $W_t$ are independent. However, that covariance can be calculated explicitly as well, if needed.

Therefore, 
\begin{equation}
    \int_0^T\E\left[\langle dx(\tau(t)),dx(\tau(t))\rangle\biggr\rvert \tau(t)\right]=\int_0^T\left(AA^Td\tau(t)+\eps^2Bx(\tau)x(\tau)^TB^Tdt\right) \ .
\end{equation}
Note that
\begin{equation}
    \E\left[\frac{1}{\delta T}\int_0^TAA^Td\tau(t)\right] = \E\left[\frac{1}{T}\int_0^T AA^Tdt\right],
\end{equation}
since the integral with $dW'_t$ is a martingale with zero mean. Therefore, if $\eps\ll\delta$ then
\begin{equation}
    \E\left[\frac{1}{\delta T}\int_0^T\langle dx(\tau(t)),dx(\tau(t))\rangle\right]\simeq\frac{1}{T}\int_0^TAA^Tdt=AA^T.
\end{equation}
$\square$
\end{proof}
If the correlations between the randomness in $\tau(t)$ and other coordinates are significant then we can evaluate Equation~\ref{cov:eqn} as follows.
\begin{lemma}
Let $\rho\in \RR^{1\times n}$ be the covariance matrix between $W'_t$ and $W_t$, which are the Brownian motions in the definitions of $\tau(t)$ and $x_t$ , 
then the quadratic variation below evaluates as:
$
    \E\left[\langle AdW_\tau, Bx(\tau)d\tau\rangle\right] = \eps A\rho x(t)^TB^T dt\ .
$
\end{lemma}
The proof follows immediately from the definition of  quadratic variation. Therefore, following a similar proof as Theorem~\ref{clock:thm}, the error to $\tilde{A}\tilde{A}^T$ from Theorem~\ref{clock:thm} is changed by an additive factor that is $O(\eps \max_{i.j}\{\rho_{ij}\})$. So when all coordinates of $\eps \rho$ are small, the characterization of $\tilde{A}\tilde{A}^T$ in Theorem~\ref{clock:thm} continues to hold.

Theorem~\ref{clock:thm} shows that the covariance matrix $AA^T$ in Algorithm~\ref{algmain} is determined by an approximate clock as long as $\frac{\epsilon}{\delta}$ is small. In the following subsections, we verify the same fact empirically.

\subsection{Empirical verification using synthetic data}\label{sec:synth}

First, we use synthetic data to plot the error in determining the matrix $AA^T$, i.e., $AA^T-\tilde{A}\tilde{A}^T$, as a function of the number of errors in the cell sample orderings. Note that the latter decreases as $\frac{\epsilon}{\delta}\to 0$. In Figure~\ref{synth:fig}, we plot the Frobenius norm of $AA^T-\tilde{A}\tilde{A}^T$ as a function of the number of errors (equal to the number of inversions) in ordering.

\begin{figure}[htb!]
    \centering
        \centering
        \includegraphics[height=2.0in]{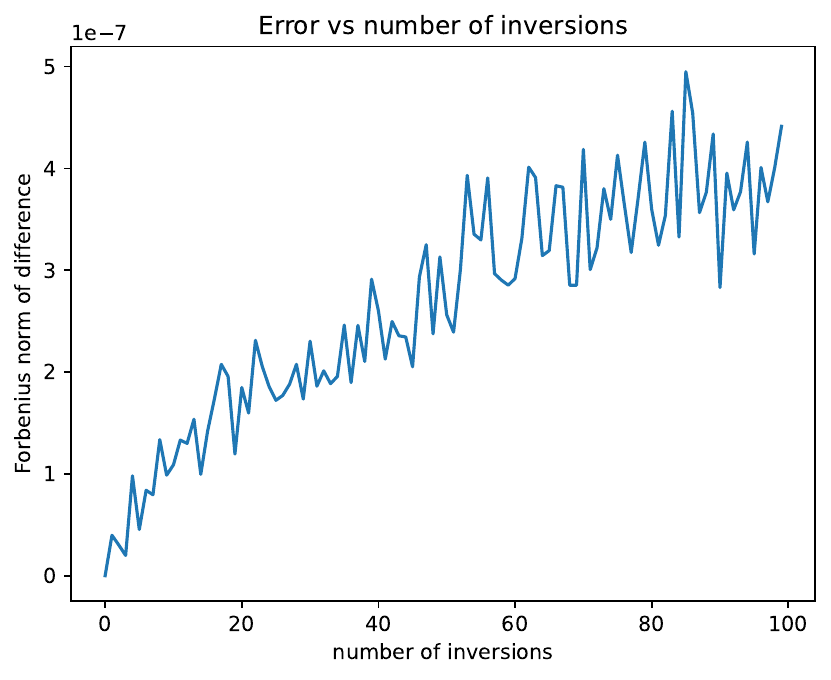}
        \caption{Error in determining covariance matrix as a function of the amount of error in the construction of the approximate clock. Note that the Frobenius norm of the difference is small and does not grow too quickly, and is close to zero near zero.}
        \label{synth:fig}
\end{figure}

In order to generate the plot in Figure~\ref{synth:fig}, we used an OU process with a $10 \times 10$ identity matrix as a stand-in for $A$; we set $B:= -Id+G$, where $G$ is a random Gaussian matrix, and we used time-steps of size $0.0001$ to generate the sample path for $1000$ time-steps. We introduce $k$ errors/inversions (the $x$ coordinate in Figure~\ref{synth:fig}) in the sample path by picking $k$ pairs of points, uniformly at random, in the sample path and swapping the pairs. Finally, we compute the corresponding empirical covariance matrix: $\tilde{A}(k)\tilde{A}(k)^T$, from the perturbed path with $k$ inversions, the corresponding Frobenius norm of $\tilde{A}(k)\tilde{A}(k)^T-AA^T$ and plot the latter as a function of $k$.

\subsection{The approximate clock using single cell RNA-seq datasets}\label{app:sec}

\begin{figure}[htb]
    \centering
    \begin{subfigure}[t]{0.5\textwidth}
        \centering
        \includegraphics[height=1.8in]{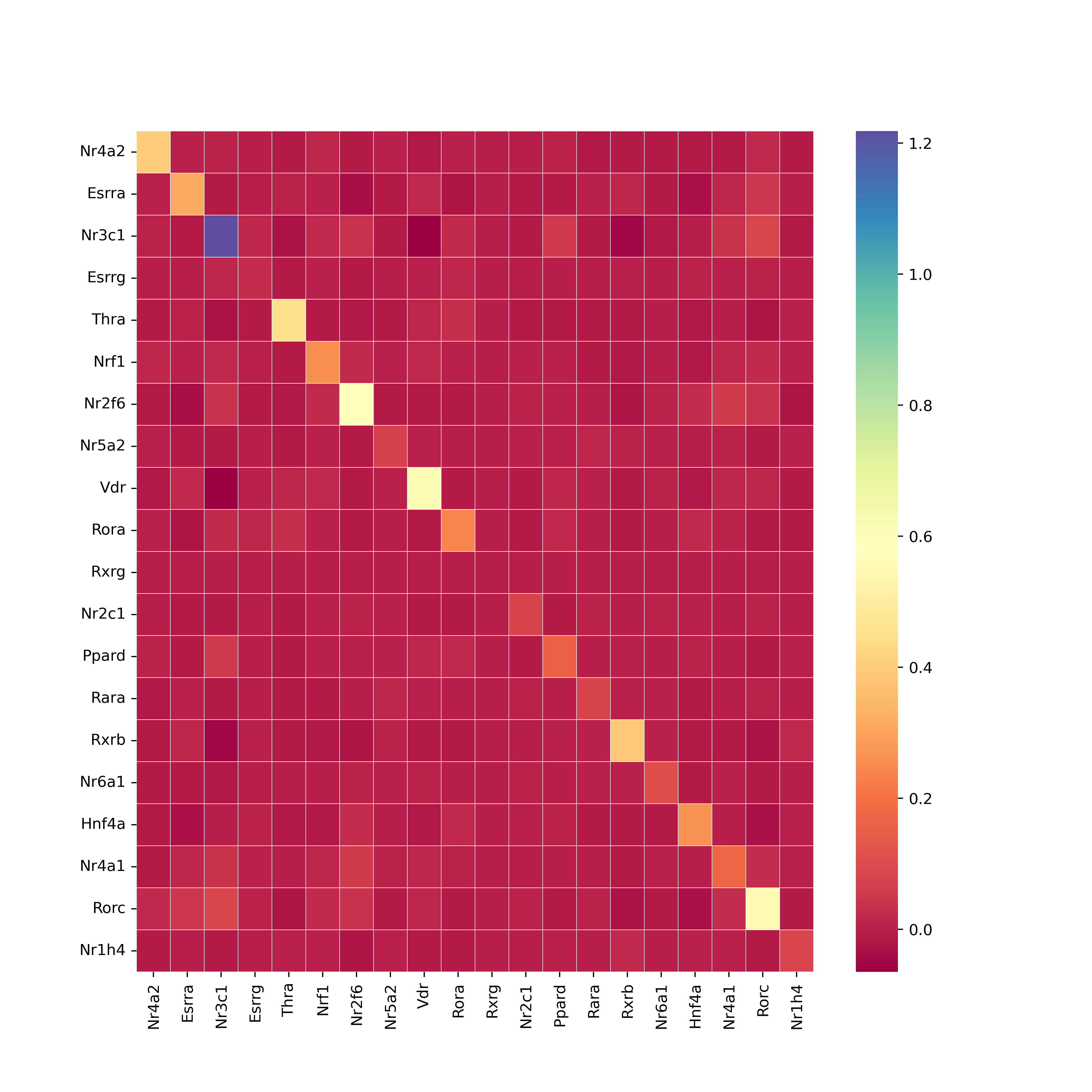}
        \caption{Covariance matrix of increments obtained from latent time computation}
    \end{subfigure}%
    ~ 
    \begin{subfigure}[t]{0.5\textwidth}
        \centering
        \includegraphics[height=1.8in]{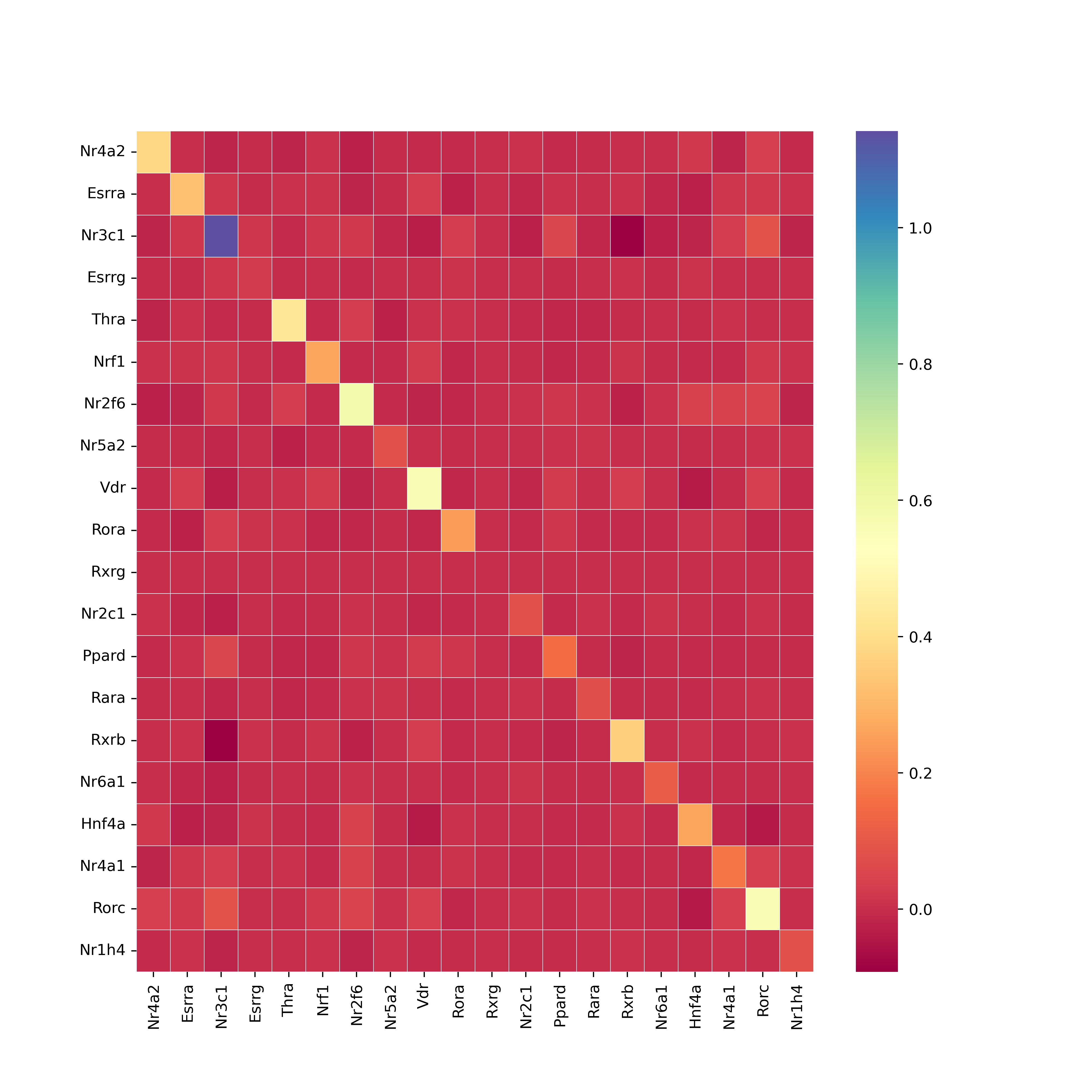}
        \caption{Covariance matrix of increments obtained from Cyclin clock}
    \end{subfigure}
    \centering
    \begin{subfigure}[b]{0.45\textwidth}
        \centering
        \includegraphics[height=2in]{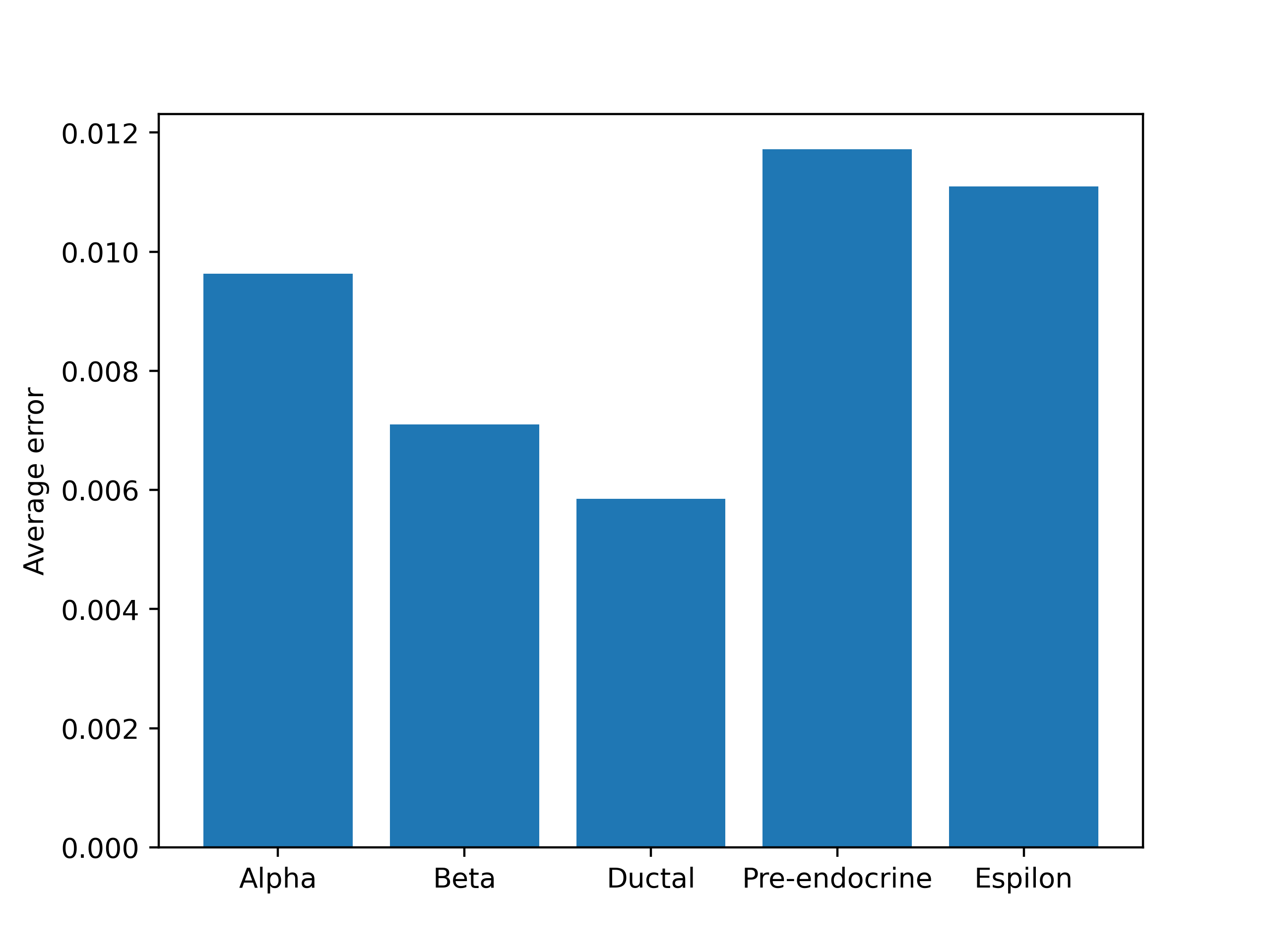}
        \caption{Average entry-wise difference between the recovered dynamical system matrix for various cell clusters}
    \end{subfigure}
    \begin{subfigure}[b]{0.45\textwidth}
        \centering
        \includegraphics[height=2in]{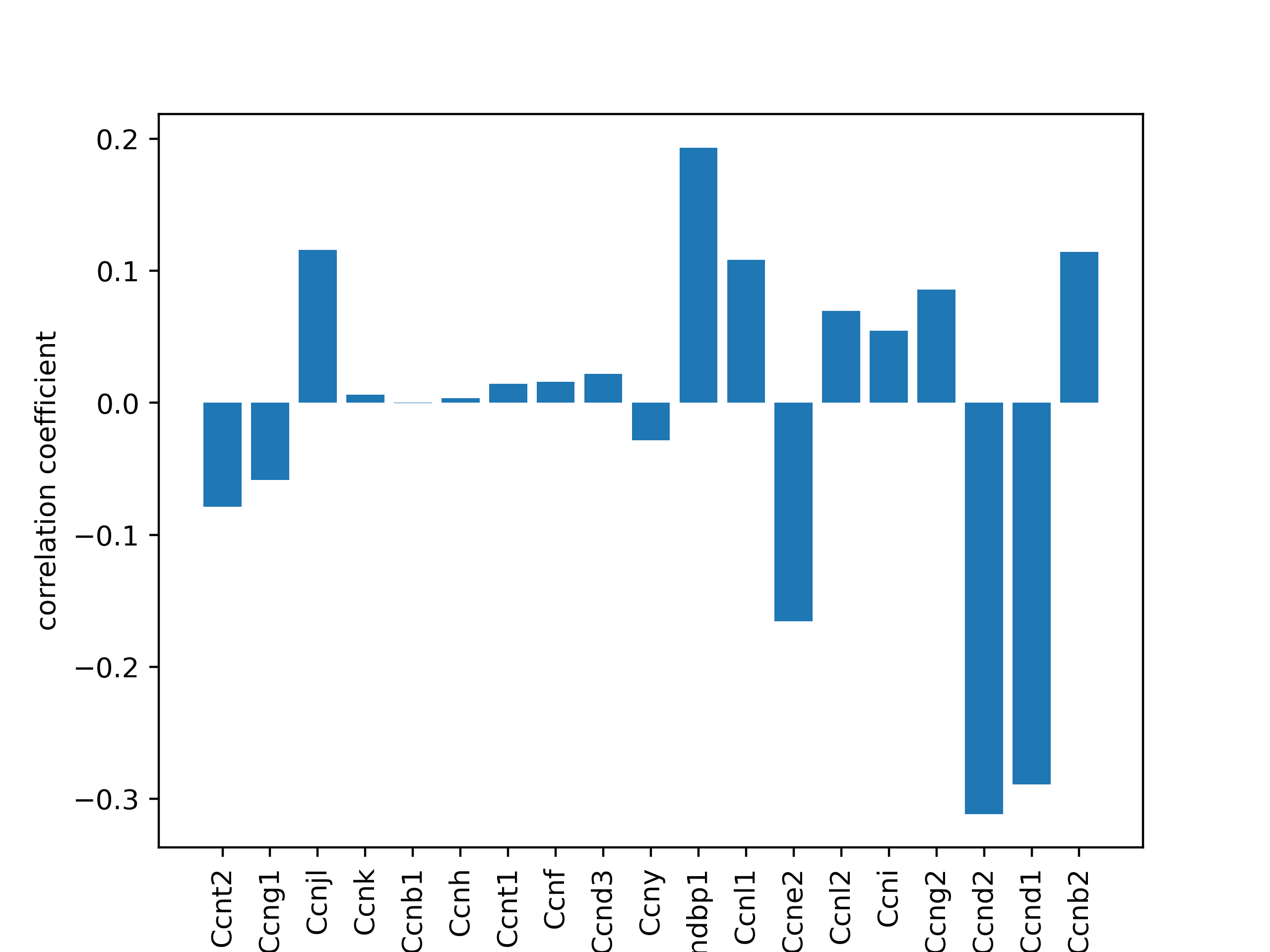}
        \caption{Correlations between cyclin expressions and scVelo latent time}
    \end{subfigure}
    \caption{Robustness of Cyclin clock in single cell pancreatic RNA-seq data}\label{scpancreas:fig}
\end{figure}


    


It can be argued that neither Theorem~\ref{clock:thm} nor Figure~\ref{synth:fig} uses cyclin expressions anywhere. Perhaps using cyclin expressions to order cells leads to a significantly large number of errors in cell sample orderings of RNA-seq data. In this subsection, we allay that fear as well. 

Yet another way to verify that the approximate clock using cyclin expressions does not lead to a large error in-between $AA^T$ and $\tilde{A}\tilde{A}^T$ is to use single cell RNA-seq data. Such datasets admit an ordering of cell samples using latent time based on underlying RNA velocity models. In this subsection, we verify that $AA^T$ (determined using latent time) and $\tilde{A}\tilde{A}^T$ (determined using a simple ordering using cyclin expressions) are close, based on single cell RNA-seq dataset from murine pancreatic tissue~\cite{Bastdas-Ponce}.


We compute the covariance matrix of increments, among significantly expressed nuclear receptors in the murine pancreatic scRNA-seq dataset from~\cite{Bastdas-Ponce}, in two ways.
First, we compute the matrix $\tilde{A}\tilde{A}^T$ using an approximate clock that involves ordering the cell samples by the difference in between two cyclin expressions, say Cyclins I and D, i.e., $\tau(c):=\mathrm{Cyclin\ I} - \mathrm{Cyclin\ D}$ for a given cell $c$. Cyclin I and D represent the optimal choice based on the correlations in Figure~\ref{scpancreas:fig}(d).
Second, we compute the matrix $AA^T$ by using the latent time function in scVelo as our ordering ( ground truth). 
We use the total RNA expression levels data for nuclear receptors in the five groups: alpha, beta, epsilon, pre-endocrine and ductal cells in the murine pancreatic single cell RNA-seq dataset~\cite{Bastdas-Ponce}. Thus we obtain five instances of $AA^T$ and 
of $\tilde{A}\tilde{A}^T$. 

As can be seen in Figure~\ref{scpancreas:fig}(c), the error measured as the average of entry-wise absolute difference of the computed matrix i.e., $AA-\tilde{A}\tilde{A}^T$ is indeed small, much smaller than the variances (diagonals in the heatmaps), thus leading credence to our claim that {\em covariance matrices of increments are robust to few errors in sample orderings}.

We performed a similar analysis on the murine dentate gyrus dataset of~\cite{dg}. The covariance matrices computed in two ways (using scVelo latent time and cyclin expression based ordering) are shown in Figure~\ref{dg-fig}.
\begin{figure}[htb]
    \centering
    \begin{subfigure}[t]{0.5\textwidth}
        \centering
        \includegraphics[height=2.25in]{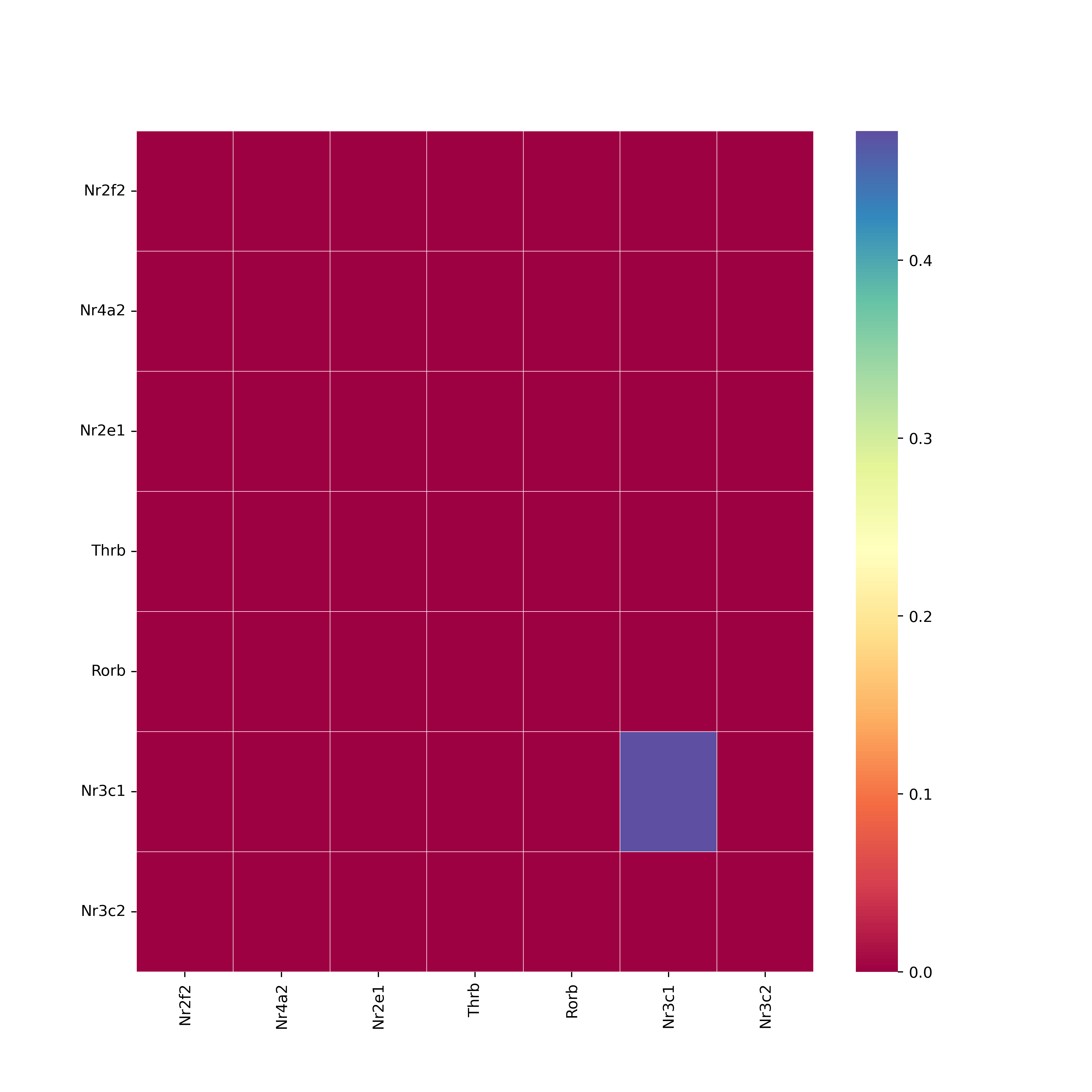}
        \caption{Covariance matrix of increments obtained from latent time computation}
    \end{subfigure}%
    ~ 
    \begin{subfigure}[t]{0.5\textwidth}
        \centering
        \includegraphics[height=2.25in]{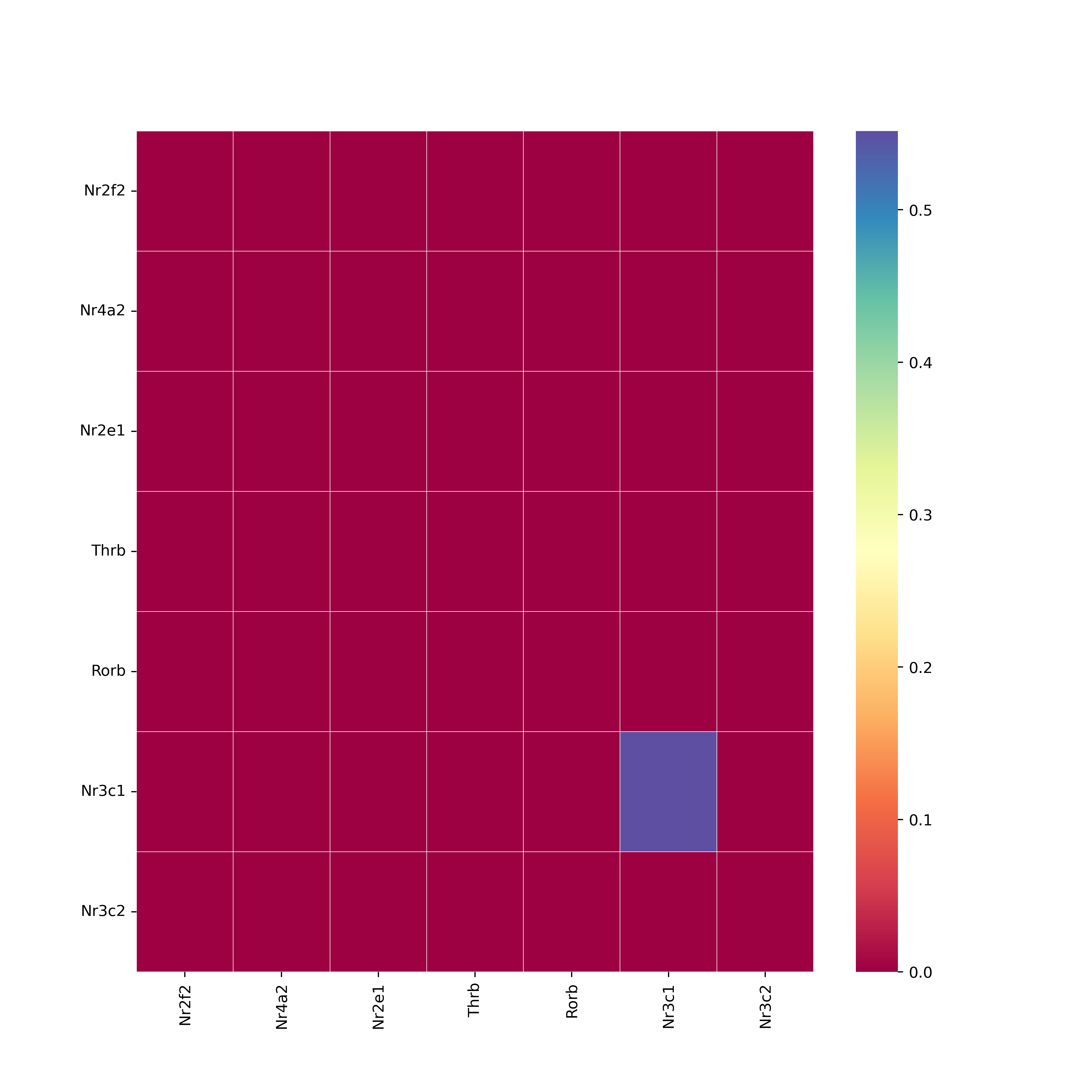}
        \caption{Covariance matrix of increments obtained from Cyclin clock}
    \end{subfigure}
    \caption{Robustness of Cyclin clock in single cell dentate gyrus RNA-seq data (note that the other nuclear receptors were not well expressed hence were excluded)}\label{dg-fig}
\end{figure}
The covariance matrices thus computed are virtually identical, again demonstrating the effectiveness of the simple cyclin based ordering of samples, as far as computation of covariance matrices is concerned. Our next step is to demonstrate the use of these covariance matrices in the computation of pathways, in Figure~\ref{paths:fig}, in the next section.
\section{Proofs from Subsection~\ref{th2:sbs}}\label{thm2:sec}


\begin{theorem}\label{basket:thm}
The objective computed by Algorithm~\ref{basket:algo} is equal in expectation to that in Equation~\ref{eqn:obj_lyp}.
\end{theorem}
\begin{proof}
First, we need to verify that the relaxation with the objective in Equation~\ref{eqn:obj_lyp2} is indeed a relaxation of Equation~\ref{eqn:obj_lyp}. Note that, Equation~\ref{eqn:obj_lyp} contains a Hadamard product with the matrix $\mathrm{Diag}(S)$, which is replaced with a Hadamard product with the positive semidefinite matrix $Y$ (with $Y_{ii}\le 1$) in Equation~\ref{eqn:obj_lyp2}. Since the former is a special case of the latter, and since the $\max$ term involve the standard SDP for computing the maximum eigenvalue, any solution to the former is also a solution for the latter. Thus, the objective in Equation~\ref{eqn:obj_lyp2} is indeed a relaxation of Equation~\ref{eqn:obj_lyp}, and its value is no greater than that of Equation~\ref{eqn:obj_lyp}.

Second, we need to verify that the optimal equilibrium matrix $Y^*$ exists for Equation~\ref{eqn:obj_lyp2}. An optimal equilibrium solution exists, because the min-max problem in Equations~\ref{eqn:obj_lyp2},~\ref{eqn:constr_lyp3}, and~\ref{eqn:constr_lyp4} can be equivalently viewed as a concave-convex pay-off of a zero-sum game with convex constraints (see for example~\cite{tes}). Therefore, by Von Neumann's min-max theorem an equilibrium solution exists. Next, for a fixed choice of optimal $X^*, Z^*$, observe that, $Y$ minimizes the eigenvalue of a symmetric matrix: $\frac{B'+B'^T}{2}\odot X -\alpha\frac{B+B^T}{2}\odot Z$. 

Third, we need to verify that strong SDP duality holds (with only minor additive error), so that the objective value and the matrix $Y^*$ from our SDP solution (Equation~\ref{eqn:obj_lyp3}) equates to the one in Equation~\ref{eqn:obj_lyp2}. In order for strong SDP duality to hold, the solution matrices $X,Y,Z$ have to be full-rank. That is not true in our problem, as they will be all rank $1$. However, perturbing them by a very small multiple of the identity matrix will make them positive definite. That will also not effect the equilibrium by more than an arbitrary small additive parameter that we can chose.

Fourth, we need to verify that $Y$ can be written as the moment matrix of first and second moments a random vector $R$ consisting of Bernoulli random variable entries, since we use a Bernoulli vector for the rounding in Algorithm~\ref{basket:algo}. This is indeed the case, as we can choose $R_i=\mathrm{Ber}(Y_{ii})$, and $R_i$ independent of $R_j$ for $i\neq j$. It is immediate that $Y = \mathbf{E}[RR^T]$, where the expectation is applied entry-wise.

Finally, we need to verify that the Bernoulli rounding in Algorithm~\ref{basket:algo} does not make the resulting approximate solution of Equation~\ref{eqn:obj_lyp} too large in expectation. Let $Y^*$ be the equilibrium solution for Equation~\ref{eqn:obj_lyp3} after rounding, and let $X^*$ and $Z^*$ be the corresponding solutions for the objective in Equation~\ref{eqn:obj_lyp2} (which is upper bounded by the value of Equation~\ref{eqn:obj_lyp}). Observe that, for fixed $X^*, Z^*$, we can write the objective as an expectation over the quadratic form in the entries of the $0$-$1$ valued vector $R\in\RR^n$. Since $R$ is $0$-$1$ solution, the expected value of this quadratic form equals at least that of Equation~\ref{eqn:obj_lyp} (which is the true $0$-$1$ valued optimum). Therefore, the expected value of the SDP solution equals that of Equation~\ref{eqn:obj_lyp}. Hence the proof follows. $\square$
\end{proof}

\clearpage
\section{Recovered interaction and perturbation matrices} 
\begin{figure}[htb!]
    \centering
    \begin{subfigure}[t]{0.45\textwidth}
        \centering
        \includegraphics[height=3.1in]{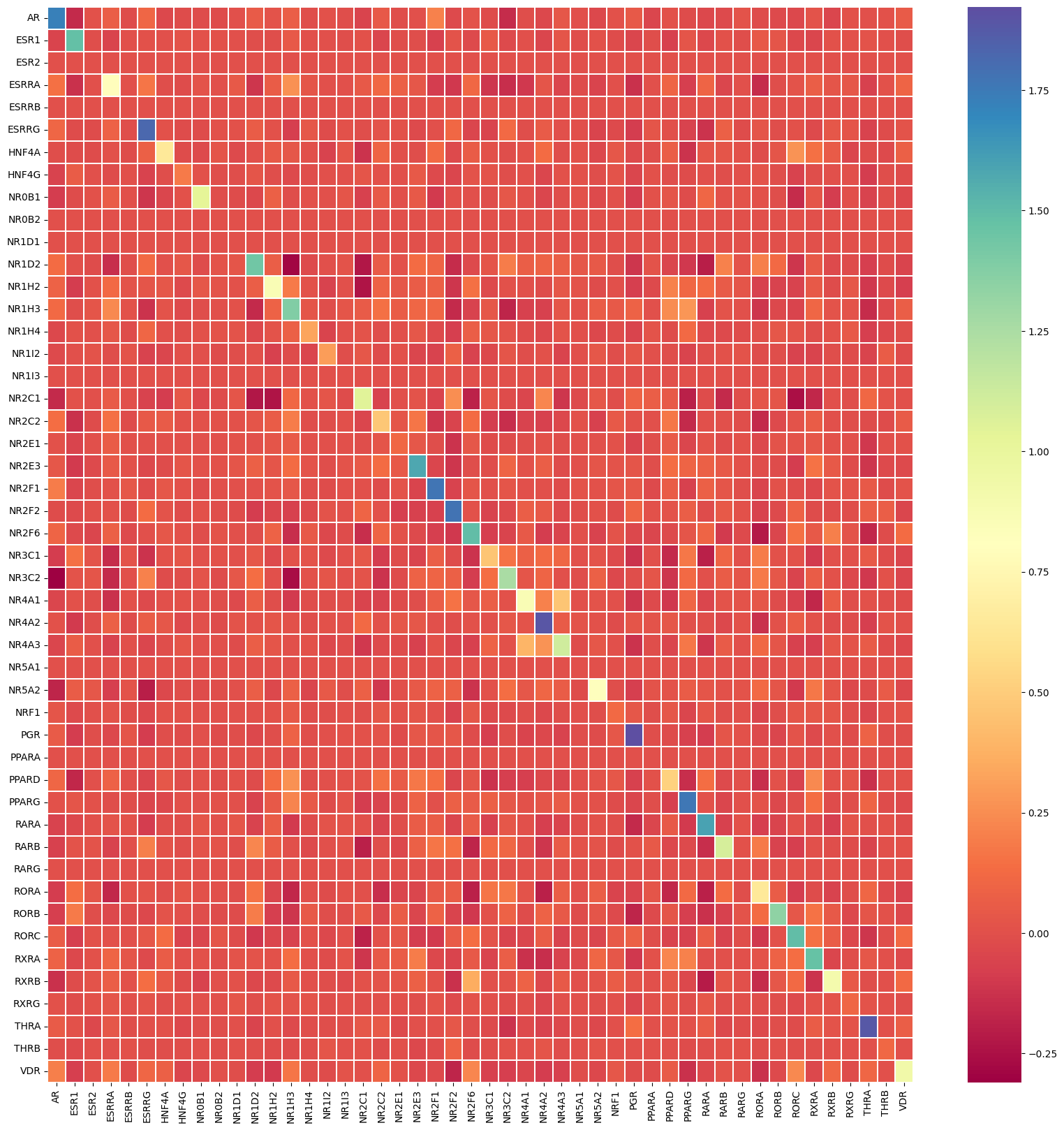}
        \caption{Underlying dynamical system for non-metastatic cases}
    \end{subfigure}\\%
    ~ 
    \begin{subfigure}[t]{0.45\textwidth}
        \centering
        \includegraphics[height=3.1in]{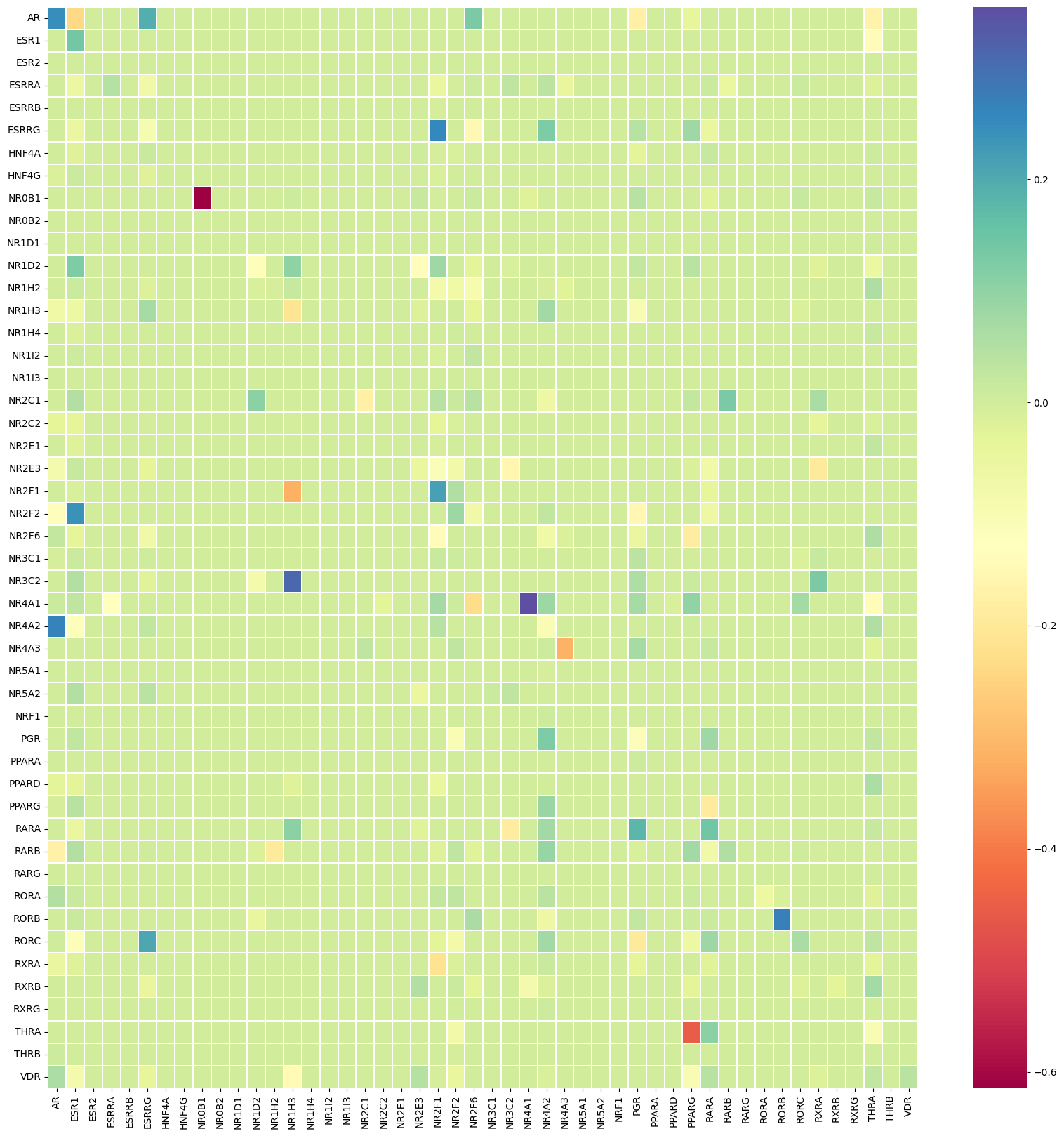}
        \caption{Perturbations of GRN towards metastasis. The matrix represents the first order approximation of changes in interaction between genes in presence of metastasis. Values far from $0$ (yellow) represent a prediction of significant up or down regulation in an interaction between a pair of genes. }
    \end{subfigure}
    \caption{Recovered dynamical system and perturbation for breast cancer dataset~\cite{metabric,metabric2,metabric3}.}\label{dyn:fig}
\end{figure}

\begin{figure}[htb!]
    \centering
    \begin{subfigure}[t]{0.45\textwidth}
        \centering
        \includegraphics[height=3.6in]{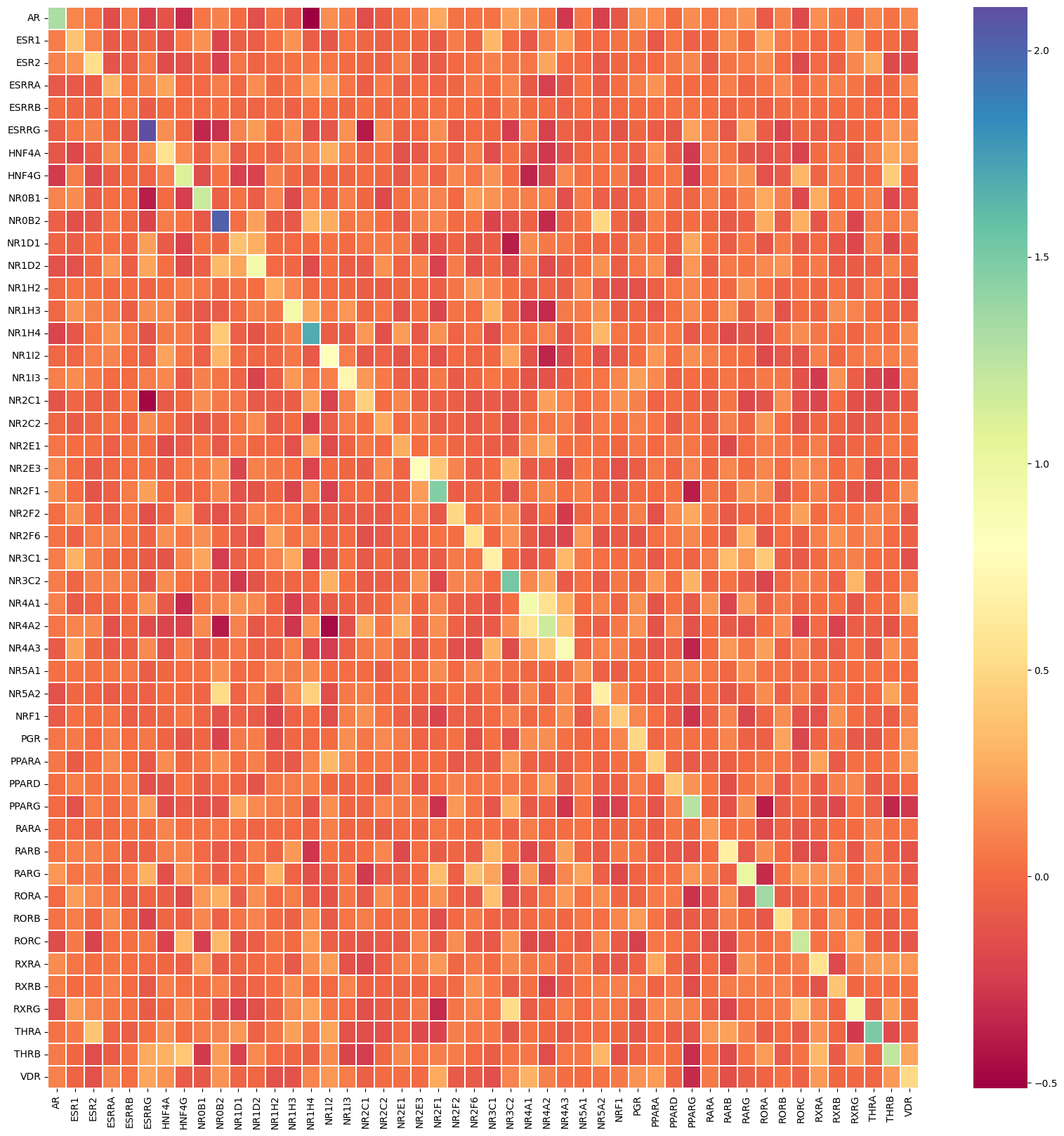}
        \caption{Underlying dynamical system for non-metastatic cases}
    \end{subfigure}\\%
    ~ 
    \begin{subfigure}[t]{0.45\textwidth}
        \centering
        \includegraphics[height=3.6in]{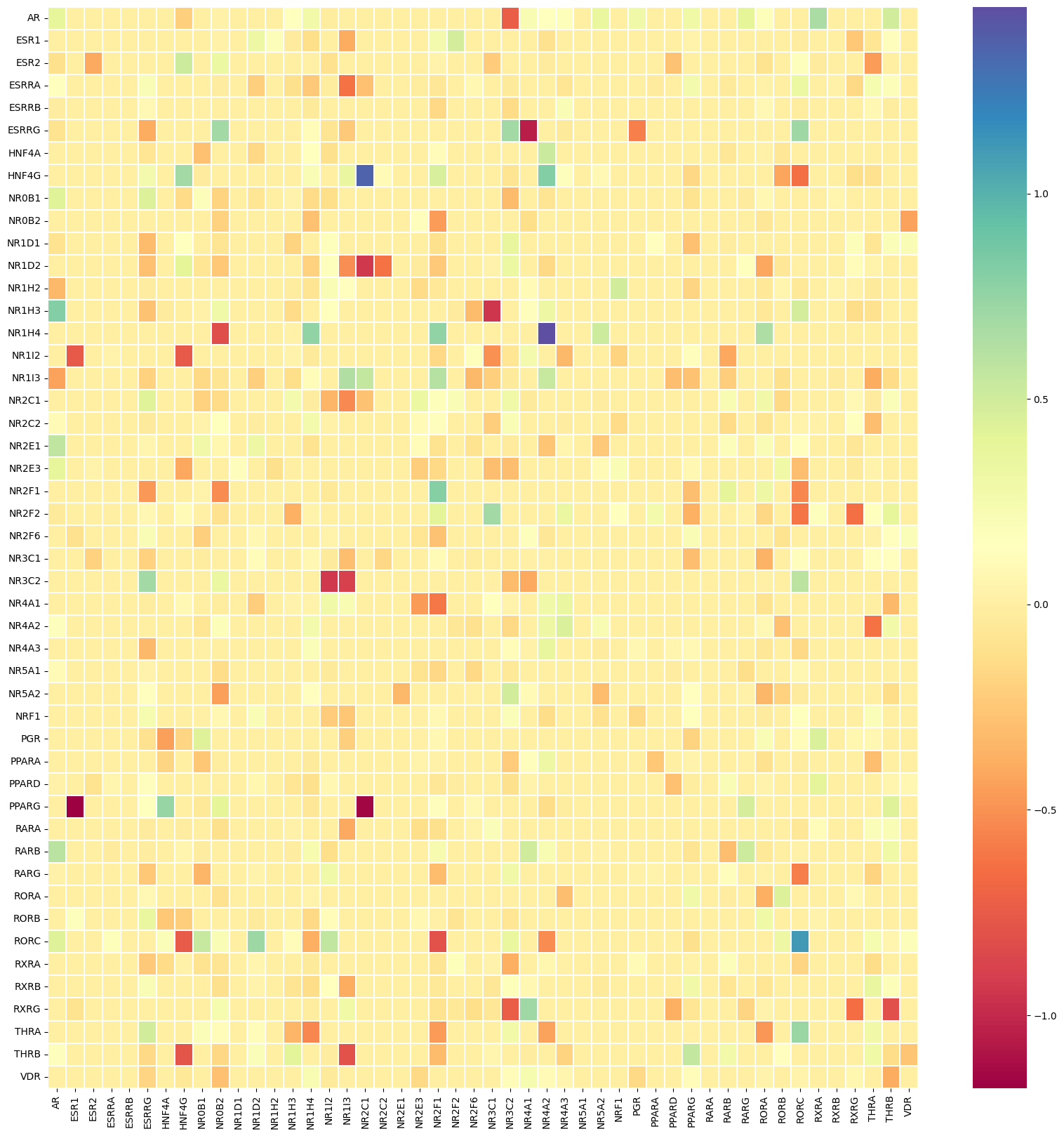}
        \caption{Perturbations of GRN towards metastasis. The matrix represents the first order approximation of changes in interaction between genes in presence of metastasis. Values far from $0$ (yellow) represent a prediction of significant up or down regulation in an interaction between a pair of genes. }
    \end{subfigure}
    \caption{Recovered dynamical system and perturbation for colorectal cancer dataset~\cite{rc}.}\label{dyn2:fig}
\end{figure}


\end{document}